\documentclass[a4paper, pdftex,pra,aps,twocolumn,twoside,nofootinbib, accepted=2021-11-07]{quantumarticle}
\pdfoutput=1
\usepackage[numbers,sort&compress]{natbib}

\usepackage[T1]{fontenc}
\usepackage[utf8]{inputenc}
\usepackage[colorlinks=true,linkcolor=blue,citecolor=red,plainpages=false,pdfpagelabels]{hyperref}
\usepackage{float}
\usepackage{braket}
\usepackage{indentfirst}
\usepackage{enumerate}
\usepackage{url}
\usepackage[dvipsnames]{xcolor}
\usepackage[stable]{footmisc}
\usepackage{fancyhdr}
\usepackage[all]{xy}
\usepackage{graphicx}
\usepackage{adjustbox}
\usepackage{tabularx}
\usepackage{youngtab}
\usepackage{mathtools}
\usepackage{enumerate}
\usepackage{textcomp}
\usepackage{amsmath,amsfonts,amssymb,amsthm}
\usepackage[mathscr]{eucal}
\usepackage{tcolorbox}
\usepackage{tabularx}

\usepackage{makecell}
\usepackage{makecell}
\usepackage[T1]{fontenc}
\usepackage[utf8]{inputenc}
\usepackage[normalem]{ulem}
\usepackage{amssymb}
\usepackage{hyperref}
\usepackage{placeins}
\usepackage{float}
\usepackage{indentfirst}
\usepackage{enumerate}
\usepackage{url}
\usepackage[dvipsnames]{xcolor}
\usepackage[stable]{footmisc}
\usepackage{fancyhdr}
\usepackage{appendix}
\usepackage{graphicx}
\usepackage[all]{xy}
\usepackage{graphicx}
\usepackage{tabularx}

\usepackage{cancel}
\usepackage{mathtools}
\usepackage{enumerate}
\usepackage{textcomp}
\usepackage{amsmath,amsfonts,amssymb,amsthm}
\usepackage[mathscr]{eucal}
\usepackage{wrapfig}
\theoremstyle{plain}
\newtheorem{theorem}{Theorem}
\newtheorem{lemma}[theorem]{Lemma}
\newtheorem{proposition}[theorem]{Proposition}
\newtheorem{corollary}[theorem]{Corollary}

\newtheorem{definition}[theorem]{Definition}

\newtheorem{notation}[theorem]{Notation}

\newtheoremstyle{note}{\topsep}{\topsep}{\slshape}{}{\scshape}{}{ }{}
\theoremstyle{note}

\newtheorem{example}[theorem]{Example}
\newtheorem*{theorem*}{Theorem}

\newcommand\tr{\operatorname{Tr}}

\newcommand{\<}{\langle}
\renewcommand{\>}{\rangle}

\newcommand\be{\begin{equation}}
\newcommand\ee{\end{equation}}
\newcommand\bea{\begin{array}}
	\newcommand\eea{\end{array}}
\newcommand\ben{\begin{eqnarray}}
\newcommand\een{\end{eqnarray}}
\newcommand\ot{\otimes}

\newcommand\bei{\begin{itemize}}
	\newcommand\eei{\end{itemize}}
\newcommand\bee{\begin{enumerate}}
	\newcommand\eee{\end{enumerate}}

\DeclarePairedDelimiter\floor{\lfloor}{\rfloor}

\begin{document}
	\title{Multiport based teleportation - transmission of a large amount of quantum information
	}
	
	\author{Piotr Kopszak$^{1}$, Marek Mozrzymas$^{1}$, Micha{\l} Studzi\'nski$^{2}$ and Micha{\l} Horodecki$^{3}$}
	\affiliation{$^1$ Institute for Theoretical Physics, University of Wrocław
		50-204 Wrocław, Poland \\
		$^2$ Institute of Theoretical Physics and Astrophysics and National Quantum Information Centre in Gda{\'n}sk,
		Faculty of Mathematics, Physics and Informatics, University of Gda{\'n}sk, 80-952 Gda{\'n}sk, Poland\\
		$^{3}$ International Centre for Theory of Quantum Technologies, University of Gda{\'n}sk, 80-952, Poland
}
	\begin{abstract}
	We analyse the problem of transmitting a number of unknown quantum states or one composite system in one go. We derive a lower bound on the performance of such process, measured in the entanglement fidelity. The obtained bound is effectively computable and outperforms the explicit values of the entanglement fidelity calculated for the pre-existing variants of the port-based protocols, allowing for teleportation of a much larger amount of quantum information. The comparison with the exact formulas and similar analysis for the probabilistic scheme is also discussed. In particular, we present the closed-form expressions for the entanglement fidelity and for the probability of success in the probabilistic scheme in the qubit case in the picture of the spin angular momentum.
	\end{abstract}
 	\maketitle	
\section{Introduction}
\label{intro}
In 2008 the novel port-based teleportation protocol (PBT) has been proposed~\cite{ishizaka_asymptotic_2008,ishizaka_quantum_2009}. In contrast to the very first teleportation procedure, discovered in~\cite{bennett_teleporting_1993}, it does not require a correction on the receiver's side depending on the classical outcome of the sender's measurement, see Figure~\ref{FPBTa}.
The lack of the correction led to various new applications, where the ordinary teleportation fails, like for example NISQ protocols~\cite{ishizaka_asymptotic_2008,Banchi2020}, position-based cryptography~\cite{beigi_konig}, fundamental limitations on quantum channels discrimination~\cite{limit}, connection between non-locality and complexity~\cite{buhrman_quantum_2016}, and many other important results~\cite{Ebler,PhysRevLett.123.210502,Stroing,sim,PhysRevA.59.156,jeong2020generalization}. 

The huge advantage of the lack of receiver's correction comes at a price. 
Due to no-programming theorem \cite{Nielsen1997} the ideal transmission in such scheme is possible only when parties exploit an infinite number of maximally entangled pairs.
\begin{figure}[h]
	\begin{centering}
		\includegraphics[width=0.3\textwidth]{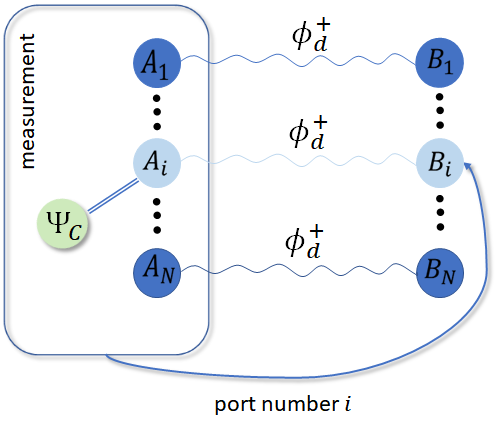}
		\caption{
		The standard configuration for the PBT scheme: two parties share $N$ copies of the maximally entangled state, called port, $\phi_d^+=|\phi^+_d\>\<\phi^+_d|$, where $|\phi^+_d\>=(1/\sqrt{d})\sum_i |ii\>$. Alice (sender) to send an unknown state $\Psi_C$ to Bob (receiver), performs a global measurement (POVM) on the states and her half of the maximally entangled pairs. As an output she gets a classical output $1\leq i\leq N$ indicating ports on the Bob's side where the state arrives. To recover the state Bob has only to pick up the right port, according to the classical message $i$ obtained from Alice,  no further correction is needed. We have also the optimised version of PBT, where Alice optimises jointly measurements and the shared states with Bob before she runs the protocol.}
		\label{FPBTa}%
	\end{centering}
\end{figure}
Accordingly, we distinguish {\it deterministic scenario}  
where teleportation  is imperfect and the state after teleportation is distorted   
and  \textit{probabilistic scenario}
where teleportation is perfect but one has to accept the non-zero failure probability of the whole process. In the first case, to learn about the efficiency we ask how the fidelity of the teleported particle depends on the number of shared entangled pairs, while in the latter we ask about the similar dependence for probability of success of perfect transmission. We can calculate the performance of the \textit{non-optimal PBT}, where parties share maximally entangled pairs, and \textit{optimised PBT}, where Alice optimises jointly over the shared state and measurements before she runs the protocol, see Figure~\ref{FPBTa}.

We rigorously address here a fundamental question of transmission capability of PBT raised firstly in~\cite{strelchuk_generalized_2013}.  Namely we ask: {\it how many qubits can one faithfully teleport for a given number of ports?}

To perform teleportation of a state of $k$ multiparty systems  one can: use the PBT with port dimension large enough, apply the PBT several times (sequential PBT), divide total number of ports into smaller packages and send every particle via such subsystem (packaged PBT), 
or finally apply {\it multi-port based  teleportation} (MPBT) suggested in~\cite{strelchuk_generalized_2013}. In particular, one could also run the recycling protocol for PBT suggested in~\cite{strelchuk_generalized_2013}, however due to recent results derived in~\cite{rec_new} the efficiency of such scheme is an open problem.
In the first case, increasing dimensionality of the port significantly reduces the entanglement fidelity~\cite{ishizaka_quantum_2009,Studzinski2017,StuNJP,MozJPA}, requiring increasing the number of used ports to compensate it~\cite{christ2018asymptotic}. In  sequential PBT and recycling protocol, we have to resign from single-shot scenario and after each round one needs to store the transmitted state, or we have to use distorted resource state, affecting the performance.
In MPBT we allow to teleport a quantum multiparty state  in one go, where each subsystem  of the teleported multiparty system ends up in one of Bob's ports, pointed by Alice's message (see Fig. \ref{FPBT}).  
The protocol thus enjoys quite a mild correction: Bob has to permute ports according to Alice's message.

The main purpose of this paper is to evaluate and compare the transmission capability of PBT and MPBT schemes. More precisely, we consider teleportation of $\sim N^\alpha$ qubits (or more generally qudits)
through $N$ ports, where $N$ grows to infinity, and analyse the quality of transmission, i.e. we ask about fidelity or probability of success. Note, in contrast with the traditional approach to 
channel capacity, where the central quantity is the asymptotic rate $k/N$, here the main objective becomes to identify  asymptotic exponents $\alpha$ for which the transmission is possible, since the mentioned rate vanishes for port based teleportation.

In any variant of either of PBT and MPBT we obtain critical-like behaviour of the quality of transmission. We identify  the critical values  $\alpha_{\rm cr}$ of exponents $\alpha$ for several variants, both for exact asymptotic values of a figure of merit and their lower bounds. Whenever the value of $\alpha$ is below the critical value $\alpha_{\rm cr}$ the values of fidelity (or probability of success, depending on the scheme) describing transmission are 1, and 0 otherwise.

We obtain {\it qualitative} difference between deterministic scheme and probabilistic one. Namely, in the deterministic scheme, even the optimal PBT scheme is outperformed by non-optimal MPBT (i.e. one based on shared maximally entangled pairs). In particular, we argue that in non-optimal deterministic MPBT, one can teleport a much larger amount of quantum information, i.e. with $\alpha_{\rm cr}=1$, in comparison to optimal port-based teleportation, where $\alpha_{\rm cr}=2/3$, and one can teleport faithfully only up to $o(N^{2/3})$ qubits.

Unlike in deterministic variant in probabilistic non-optimal MPBT the scaling is the same as for probabilistic optimal PBT, allowing for teleportation $o(N^{\alpha})$ qubits with $\alpha_{\rm cr}=1/2$. However, considering an optimal version of probabilistic MPBT, one can transmit $o(N^{\alpha})$ with a critical exponent equal to 1, clearly outperforming PBT variants.

To achieve our results, we first provide a lower bound on fidelity in deterministic non-optimal MPBT based on the state discrimination problem. The new bound is effectively computable and depends only on global parameters like the number of ports, their dimension, and the number of teleported particles. The first bound of such kind has been discussed in~\cite{strelchuk_generalized_2013} in not fully rigorous way, suffering from some flaws discussed later in this paper.
In the qubit case, starting from group-theoretical results for exact values of entanglement fidelity and probability of success in non-optimal MPBT, presented in companion paper~\cite{Stu2020}, we deliver exact and effectively computable expressions for these quantities in qubit case, phrased in appealing form  of spin angular momentum.
Next, result regarding asymptotic behaviour of non-optimal probabilistic case, have been obtained by combining advanced  tools from statistical analysis, in particular  non-straightforwardly the celebrated Berry-Essen theorem~\cite{endriu,essen0} with  direct estimates of binomial expressions by their Gaussian approximations. For optimal probabilistic MPBT we use for our analysis the exact formula for probability of success  in such optimal MPBT derived in~\cite{2020OPT}.
In the Table~\ref{tab:schemes} we summarise the already mentioned variations of the architecture of PBT protocols, which are studied in this paper and are presented in more detail in the following sections.
\begin{table}[t]
	\centering
	\begin{tabularx}{0.5\textwidth}{>{\hsize=0.75\hsize}X >{\hsize=0.95\hsize}X >{\hsize=1.55\hsize}X | >{\hsize=0.75\hsize}X}
		Acronym & Name & Description & Reference\\
		\hline\hline
		$(O)PBT$ & (Optimal) port-based teleportation & Teleportation of one qudit using $N$ entangled pairs, possibly with optimization over states and POVMs employed (optimal protocol) & qbit case: 
		\cite{ishizaka_asymptotic_2008, ishizaka_quantum_2009}, qdit case:\cite{Studzinski2017, StuNJP}\\
		\hline
		$(O)MPBT$  & (Optimal) multi port-based teleportation & The teleportation of $k$ qudits in one go using $N$ entangled-pairs, includes the correction procedure in the last step, namely the permutation of the ports, see Fig.~\ref{FPBT} & this paper; companion papers \cite{Stu2020, 2020OPT}\\
		\hline
		\makecell[tl]{\textit{Pack.}\\\textit{(O)PBT}} & Packaged (optimal) port-based teleportation & Teleportation of $k$ qdits by using $k$ (O)PBT schemes consisting of $N/k$ entangled-pairs, see Fig.~\ref{fig:packaged} & \cite{strelchuk_generalized_2013}\\
		\hline
	\end{tabularx}
		\caption{Summary of different variations of Port Based Teleportation protocols discussed in this work. In each of them both deterministic and probabilistic scheme can be realised by choosing the appropriate measurement.}
	\label{tab:schemes}
\end{table}
\section{The Multi-port-based Teleportation}
 \label{multPBT}
In order  to transmit $k$-system state $\Psi_C=\Psi_{C_1C_2\dots C_k}$ (see Fig.~\ref{FPBT}) Alice applies a global measurement $\Pi_{\mathbf{i}}^{AC}$, where $A=A_1\dots A_N$, on her halves of state $\Phi_{AB}^+= \bigotimes_{j=1}^N |\phi^+_d\>\<\phi^+_d|_{A_jB_j}$ and the state to be  teleported.
As an output she receives a tuple of indices 
$\mathbf{i}=\{i_{1},i_{2},\ldots,i_{k}\}$ 
and sends it to Bob through a classical channel.   We denote the set  of all outputs $\mathbf{i}$ by $\mathcal{I}$. Note that the number of outputs $k!\binom{N}{k}$ grows polynomially in $N$ at fixed $k$. The meaning of  indices $(i_1, \ldots ,i_k)$  is that they point the ports on the receiver's side on which the teleported systems appear, i.e. the first system arrives at port $i_1$, the second at $i_2$ and so on. Bob recovers the initial form of the teleported state 
by suitably permuting his systems. 
 \begin{figure}[h]
 	\begin{centering}
 		\includegraphics[width=0.35\textwidth]{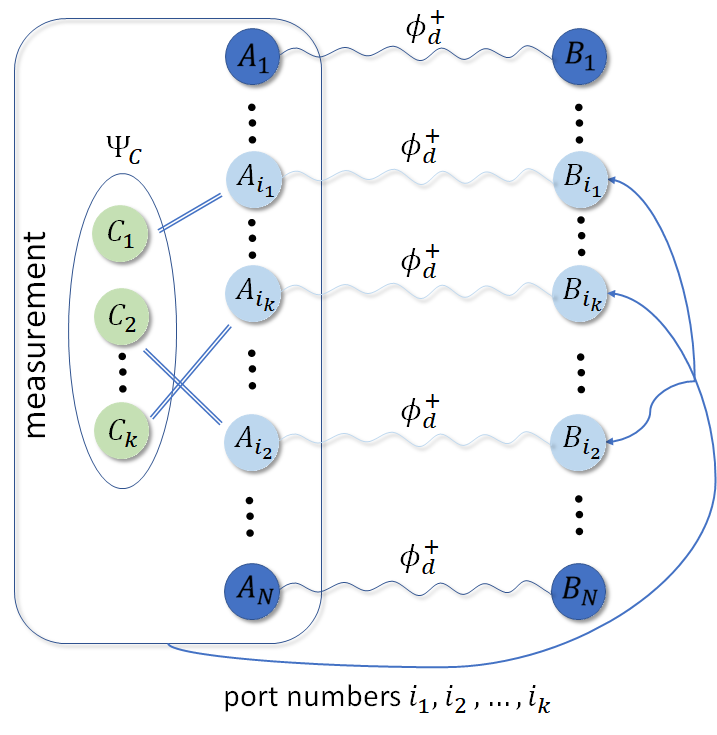}
 		\caption{ In MPBT two parties share $N$ copies of the maximally entangled state (port) $\phi_d^+=|\phi^+_d\>\<\phi^+_d|$, where $|\phi^+_d\>=(1/\sqrt{d})\sum_i |ii\>$. Alice to transmit a multipartite state $\theta_C$ of $k$ systems
 		 performs a global measurement (POVM) on the state $\Psi_C$ and her halves of the states $\phi^+_d$. As an output she gets a string $\mathbf{i}=\{i_{1},i_{2},\ldots,i_{k}\}$ indicating ports on the Bob's side where the states arrive. To recover the state Bob has to pick up pointed ports in the right order, which is equivalent to permuting the ports according to message $\mathbf{i}$. Similarly as it is for PBT here Alice can also jointly optimise over shared states and measurements getting optimal MPBT protocol.}
 		\label{FPBT}%
 	\end{centering}
 \end{figure}
The resulting action of the teleportation channel $\mathcal{N}$ is of the form
 \begin{eqnarray}
 \label{ch1}
	 &&\mathcal{N}\left(\Psi_{C} \right)=\sum_{\mathbf{i}\in \mathcal{I}}\tr_{A\bar{B}_{\mathbf{i}}C}\left[ \sqrt{\Pi_{\mathbf{i}}^{AC}}\left(\Phi_{AB}^+\ot \Psi_{C} \right)\sqrt{\Pi_{\mathbf{i}}^{AC}}^{\dagger}\right] \nonumber \\
	&&=\sum_{\mathbf{i}\in \mathcal{I}} \tr_{AC}\left[\Pi_{\mathbf{i}}^{AC}\left( \sigma_{\mathbf{i}}^{AB}  \ot \Psi_{C}\right)\right],
 \end{eqnarray}
where the bar in $\bar{B}_{\mathbf{i}}$ denotes discarded  subsystems except those on positions 
$i_{1},i_{2},\ldots,i_{k}$.
The states $\sigma_{\mathbf{i}}^{AB}$ or shortly $\sigma_{\mathbf{i}}$ for $\mathbf{i}\in \mathcal{I}$, called later the {\it signals}, are given as
\begin{equation}
\label{sigma}
\begin{split}
\sigma_{\mathbf{i}}^{AB}&\equiv \tr_{\bar{B}_{\mathbf{i}}}\Phi_{AB}^+=\frac{1}{d^{N-k}}\mathbf{1}_{\bar{A}_{\mathbf{i}}}\ot \phi^+_{A_{\mathbf{i}}B_{\mathbf{i}}},
\end{split}
\end{equation}
where $\phi^+_{A_{\mathbf{i}}B_{\mathbf{i}}}=|\phi_d^+\>\<\phi_d^+|_{A_{\mathbf{i}}B_{\mathbf{i}}}$.
To examine the efficiency we compute the entanglement fidelity which reports how well parties can transmit quantum correlations by sending $k$ halves of maximally entangled states $\Phi^+_{CD}=\phi_{C_1D_1}^+\ot \phi_{C_2D_2}^+\ot \cdots \ot \phi_{C_kD_k}^+$: 
\begin{equation}
\label{ent_fid}
\begin{split}
	F&=\tr\left[\Phi^+_{BD}(\mathcal{N}\ot \mathbf{1}_{D})\left(\Phi^+_{CD}\right) \right]=\frac{1}{d^{2k}}\sum_{\mathbf{i}\in\mathcal{I}}\tr\left[\Pi_{\mathbf{i}}^{AB} \sigma_{\mathbf{i}}^{AB}\right].
\end{split}
\end{equation}
In the probabilistic scheme the teleportation channel looks exactly as in~\eqref{ch1}, but with different form of measurements. The channel is now trace non-preserving, since there is a POVM, denoted as $\Pi_{\mathbf{0}^{AB}}$, corresponding to the failure of the whole process. The average probability of success  reads~\cite{ishizaka_asymptotic_2008,Stu2020}:
\begin{equation}
\label{pav}
p_{succ}=\frac{1}{d^{N+k}}\sum_{\mathbf{i}\in\mathcal{I}}\tr\left[\Pi_{\mathbf{i}}^{AB}\right].
\end{equation}
The goal here is to find the set of optimal measurements, maximising the probability of success. This is done in~\cite{Stu2020}  by exploiting symmetries exhibit in the problem and methods from semidefinite programming. 
Our considerations look similar to the original PBT scheme, but incorporating $k$ particles into the teleportation process  makes it structurally entirely different, see~\cite{Stu2020}.
 The structure of the signals $\sigma_{\mathbf{i}}^{AB}$
  and the suitable measurements $\Pi_{\mathbf{i}}^{AB}$, makes the problem of evaluating $F$ and $p_{succ}$ very hard and technical. 
  Therefore, in the next section, we present an effectively computable lower bound on the entanglement fidelity. 
 
\section{Fidelity bound from the state discrimination task}
\label{Fbound}
Evaluation of fidelity in the standard PBT was a formidable task, requiring machinery of representation theory of $SU(2)^{\ot N}$ for qubits, and much more advanced representation theoretic tools for $d>2$. 
The resulting formulas are usually not very transparent, expressed in terms of complicated sums of representation theory parameters, for which no explicit expressions are known beyond qubits. 

The first attempt to derive the efficiency of deterministic MPBT has been presented in~\cite{strelchuk_generalized_2013}.  
The combinatorial argumentation used by the authors was incorrect, neglecting deeper complexity of the problem.  However,  one can save part of argumentation and 
show that the bound given in 
\cite{strelchuk_generalized_2013}
is indeed a legitimate lower bound for fidelity,
see discussion after example \ref{ex:contr} in Appendix~\ref{AppB}.
Here, we go further and give a simple and stronger lower bound for fidelity of the protocol. 
The starting point is the idea presented~\cite{ishizaka_asymptotic_2008,ishizaka_quantum_2009}, exploited in
\cite{beigi_konig,strelchuk_generalized_2013},  of providing lower bound for the standard PBT protocol 
by relating fidelity of teleportation to probability of success in state discrimination.  Let us emphasize, that generalizing this approach to the  multiport scheme case 
requires solving a complex combinatorial problem, contained in Lemma \ref{square}.

The relation between the entanglement fidelity and the averaged probability  $p_{dist}$ of distinguishing the signals $\sigma_{\mathbf{i}}$ with equal prior probability $1/(k!\binom{N}{k})$ for arbitrary measurements $\Pi_{\mathbf{i}}$ is given by (see \cite{beigi_konig}):
\begin{equation}
\label{2eqs}
F=\frac{k!\binom{N}{k}}{d^{2k}}p_{dist},\qquad p_{dist}=\frac{1}{k!\binom{N}{k}}\sum_{\mathbf{i}\in\mathcal{I}}\tr(\Pi_{\mathbf{i}}\sigma_{\mathbf{i}}).
\end{equation}
Now, our goal is to provide effectively computable lower bound on $p_{dist}$, which gives us the lower bound on $F$  of our protocol.
Since we are interested in any feasible lower bound, for any $k\geq 1$, in~\eqref{2eqs} we take  square-root measurements (SRM) of the form:
\begin{equation}
\label{srm}
\Pi_{\mathbf{i}}=\rho^{-1/2}\sigma_{\mathbf{i}}\rho^{-1/2} \quad \text{with} \quad \rho=\sum_{\mathbf{i}\in\mathcal{I}}\sigma_i.
\end{equation}
The support of square-root measurements is always restricted to the support of the signals $\sigma_{\mathbf{i}}$~\cite{ishizaka_asymptotic_2008,Stu2020} making explicit calculations difficult. 
 However, by formulating the generalisation of the statement of Lemma A.3 from~\cite{beigi_konig} (see also Appendix~\ref{AppB} for details), we get a general lower bound on $p_{dist}$:
\begin{equation}
\label{bounda}
p_{dist} \geq \frac{1}{d^{N-k}\tr \overline{\rho}^2},
\end{equation}
where $\overline{\rho}$ is a normalised version of the operator $\rho$ from~\eqref{srm}. 
To use~\eqref{bounda} one has to evaluate $\tr \overline{\rho}^2$.  The result (Appendix~\ref{AppB}) is the following:
\begin{lemma}
\label{square}
For the  operator $\overline{\rho}=\rho/\tr \rho $, we have
\begin{equation}
\tr(\overline{\rho}^2)=d^{-N-k}\binom{N}{k}^{-1}\binom{d^2+N-1}{k}.
\end{equation}
\end{lemma}
Now we are in position to provide our lower bound for teleportation fidelity.
\begin{theorem}
	\label{bound_thmF}
The entanglement fidelity $F$ in MPBT	deterministic scheme, with $N$ ports of dimension $d$ each, while teleporting $k\leq \floor{N/2}$ particles satisfies
\begin{equation}
\label{f2}
\begin{split}
F\geq\binom{N}{k}\binom{d^{2}+N-1%
}{k}^{-1}
\geq
 \left( 1-\frac{d^{2}-1}{d^{2}+N-k}\right) ^{k}.
\end{split}
\end{equation}

For fixed $k$ the fidelity scales at least as $1 - O(1/N)$ and for $N\to \infty$ the fidelity goes to $1$. 
\end{theorem}
 The first bound in~\eqref{f2} follows from Lemma~\ref{square} and expression~\eqref{2eqs}.
 The proof of second estimate, being  technical, we delegate to 
 Appendix~\ref{AppB}. The last sentence of the Theorem  is obtained by applying the Bernoulli inequality $(1-x)^k\geq 1-k x$ (valid for $x<1$)
to the second bound. In particular, for qubits it reads:
\begin{equation}
\label{reduced_bound}
F\geq 1-\frac{3k}{4+N-k}.
\end{equation}
For $k=1$, the bound in~\eqref{f2} reduces to the bound for the standard non-optimal PBT 
\cite{beigi_konig}, namely $F\geq 1-\frac{d^2-1}{d^2+N-1}$. 

Now we compare performance of MPBT with the port based teleportation protocols by considering packaged PBT~\cite{strelchuk_generalized_2013}. We shall compare MPBT protocols only with the mentioned versions of PBT, since the versions with large port dimension $d^k$ become ineffective very fast~\cite{christ2018asymptotic} and fidelity drops substantially. For our considerations we use both the original, i.e. non-optimal PBT protocol (Figure~\ref{FPBTa}), where Alice and Bob share $N$ maximally entangled states, as well as the {\it optimized} version of PBT scheme (OPBT), where Alice optimises over shared state and measurements. Let us emphasize that in our scheme we do not perform any optimisation as it is done in OPBT case and, moreover, we compare with the lower bound in Theorem~\ref{bound_thmF}.

Let us denote by $F(N,k)$ the fidelity of teleportation, regardless of the scheme, of $k$ qubit systems through $N$ ports.

\textit{Packaged PBT} In this version Alice to transmit $k-$system state to Bob has to divide all the ports in her possession in $N/k$ packages (see Figure~\ref{fig:packaged}). Then she runs independently $k$ separate PBT protocols with $N/k$ ports each. The total fidelity $F_{pack}(N,k)$ in the packaged version of PBT equals
\begin{equation}
\label{packaged}
F_{pack}(N,k):=F(N/k,1)^k.
\end{equation}
Please notice that optimal MPBT protocol is at least good as packaged PBT, since optimisation over the resource state includes all packaged schemes. Here however, we compare packaged versions of PBT with lower bound from Theorem~\ref{bound_thmF}, so it could happen in principle that the lower bound performs worse than packaged PBT. 

In paper~\cite{strelchuk_generalized_2013} it was shown that in case of qubits the quantity $F_{pack}(N,k)$ can be bounded from  below as
\begin{equation}
\label{pack_bound}
F_{pack}(N,k)=\left(1-\frac{3k}{4N}\right)^k\geq 1-\frac{3k^2}{4N}.
\end{equation}
However, it is easy to check that bound given through Theorem~\ref{bound_thmF} outperforms bound~\eqref{pack_bound} for $k\geq 4$. This motivates us to compare bound~\eqref{f2} with exact values of the packaged PBT given by~\eqref{packaged}.
\begin{figure}[h]
	\begin{centering}
		\includegraphics[width=0.3\textwidth]{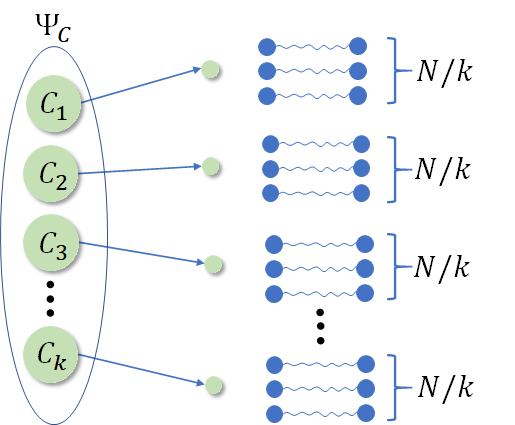}
		\caption{
		The packaged protocol for PBT. Alice to send a $k-$system state $\Psi_{C_1C_2\ldots C_k}$ to Bob divides all the ports into $N/k$ packages, and performs $k$ independent PBT protocols with $N/k$ ports each.}
		\label{fig:packaged}%
	\end{centering}
\end{figure}
We know from~\cite{ishizaka_quantum_2009}, that the fidelity for qubit optimal PBT is given by $F=\operatorname{cos}^2\left(\frac{\pi}{N+2}\right)$. Then in the packaged version entanglement fidelity reads 
\be 
\label{eq:pack_opbt}
F_{pack}^{OPBT}(N,k)=\operatorname{cos}^{2k}\left(\frac{\pi}{N/k+2}\right).
\ee

In Figure~\ref{FPBT2} we present comparison of the first bound from~\eqref{f2}, plotted for various $k$ versus the fidelity of the packaged OPBT~\eqref{eq:pack_opbt}.

In our comparisons we use only the lower bound for non-optimal MPBT from Theorem~\ref{bound_thmF}, getting regions for which we perform better. For the exact values we obviously perform even better. In Appendix~\ref{AppC} we present an explicit expression for the entanglement fidelity in the qubit case, using angular momentum representation, and we compare it with the second bound from Theorem~\ref{bound_thmF}.

\begin{figure}[h]
	\begin{centering}
		\includegraphics[width=0.45\textwidth]{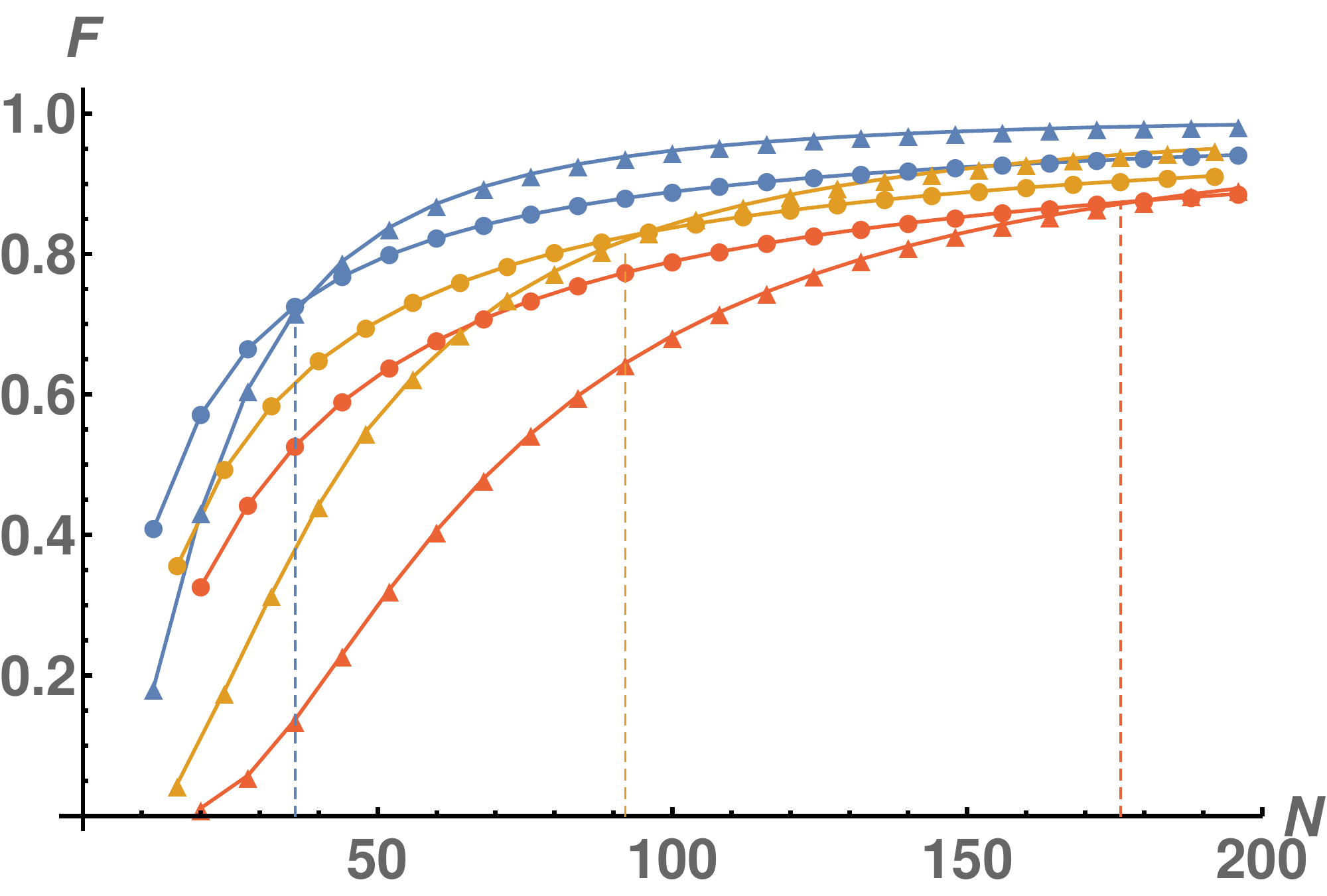} 
		\caption{Comparison of the first bound from~\eqref{f2} from Theorem~\ref{bound_thmF} (full circles, for $k=4$ marked with \protect\raisebox{-.5mm}{\protect\includegraphics[height=3mm]{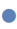}}, $k=6$ \protect\raisebox{-.5mm}{\protect\includegraphics[height=3mm]{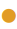}}, $k=8$		\protect\raisebox{-.5mm}{\protect\includegraphics[height=3mm]{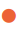}}) with exact values of the Pack. OPBT~\eqref{eq:pack_opbt} (triangles, for fixed $k=4$ \protect\raisebox{-.5mm}{\protect\includegraphics[height=4mm]{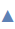}}, $k=6$ \protect\raisebox{-.5mm}{\protect\includegraphics[height=4mm]{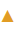}}, $k=8$ \protect\raisebox{-.5mm}{\protect\includegraphics[height=4mm]{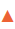}}).  For  each fixed $k=4,6,8$ there exists $N$ for which the packaged version of OPBT outperforms the bound. 
		We see, however, that for increasing number $k$ of teleported qubits, the region of supremacy of MPBT scheme increases quickly. }
		\label{FPBT2}%
	\end{centering}
\end{figure}
For fixed $k$ and small number of ports $N$ our first bound outperforms packaged OPBT from~\eqref{packaged}. However, for every fixed $k$ there is some point after which $F_{pack}^{OPBT}(N,k)$ dominates over the first (stronger) bound given by (\ref{f2}). Nevertheless, as $k$ increase this point shifts more rapidly. This suggests, that changing $k$ adaptively with growing $N$ shows the real advantage of the new protocol over the preexisting ones. We examine it in Section~\ref{adapt}.

\section{Multiport versus standard PBT: large number of teleported qubits}
\subsection{Deterministic protocol}
\label{adapt}
We  now compare performance of the MPBT protocol with the packaged (O)PBT ones in the asymptotic regime, when we change the number of teleported particles adaptively. Namely, we shall assume that the number of teleported particles scales as $k(N)=aN^{\alpha}$.  

For the case of MPBT we can formulate
\begin{proposition}
	\label{funN}
	When $k/N\to 0$,
	the fidelity approaches $1$.
	Moreover, for $k=\lfloor a N \rfloor $ 
	where $a<1$ we have  $\lim_{N\to \infty} F\geq
	   \operatorname{e}^{-(d^2-1)\frac {a}{1-a}}.$ 
\end{proposition} 
 When $k=\lfloor a N \rfloor $ the statement of the Proposition follows directly from the calculation of the limit of the first bound given in~\eqref{f2}. 
The case $k=o(N)$  follows from applying Bernoulli inequality $(1-x)^k\geq 1-k x$ (valid for $x<1$)
to the second bound in~\eqref{f2}.


The relation between $\alpha$ and the asymptotic fidelity of PBT and OPBT protocols is as follows. For the non-optimal PBT protocol the fidelity fulfils $F\approx1-\frac{3}{4N}=1-O(1/N)$~\cite{ishizaka_quantum_2009}.
Thus, we can easily calculate $\lim_{N\rightarrow\infty}F^{PBT}_{pack}(N,k)$ for the non-optimal packaged PBT. It turns out that depending on whether $\alpha$ is greater, smaller or equal to particular value $\alpha_{cr}$  (which in this case equals $1/2$) the asymptotic fidelity equals $0, 1$ or a fixed value $F_{cr}$. 

For $F_{pack}^{OPBT}(N,k)$ we can apply exactly the same reasoning, since in that case $F=\operatorname{cos}^2\left(\frac{\pi}{N+2}\right)\approx 1-{\pi^2}/{N^2}$ \cite{ishizaka_quantum_2009}.

The results for all three schemes are presented in the Table~\ref{tab:fid_alpha} and the convergence for OPBT and MPBT protocols is depicted in the Figure~\ref{fig:asym_f}.

\begin{table}[]
   \centering
    \resizebox{8.5cm}{!}{%
    \begin{tabular}{c|c|l|c}  & $F=0$ & $\qquad \qquad F_{\rm cr}$ & $F=1$  \\
    \hline
	    $Pack. PBT$   & $\alpha>1/2$ & $\alpha_{\rm cr}=1/2\,$, $F_{\rm cr}= e^{-3a^{2}/4}$ & $\alpha<1/2$  \\
	    $Pack. OPBT$   & $\alpha>2/3$ & $\alpha_{\rm cr}=2/3\,$, $F_{\rm cr}= e^{-\pi^2 a^{3}}$ & $\alpha<2/3$  \\
	    $MPBT$   & $-$ & $\alpha_{\rm cr}=1,\ \  \ $  $F_{\rm cr}\geq e^{-\frac{3a}{1-a}}$ & $\alpha<1$  \\
      
         \hline
    \end{tabular}
    }
    \caption{ Comparison of the asymptotic behaviour of two variants of packaged PBT with MPBT in deterministic version, where $k=a N^\alpha$. 
    By "cr" we denote the critical values of parameter $\alpha$ for which 
    asymptotic value of $F$ exhibits a jump.
    }
    \label{tab:fid_alpha}
\end{table}

\begin{figure}[!]
	\begin{centering}
		\includegraphics[width=0.45\textwidth]{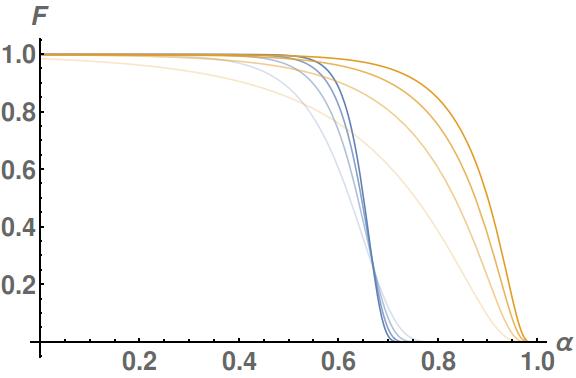}
		\caption{ 
		Asymptotic values of the lower bound on entanglement fidelity in MPBT  (orange lines) compared to the fidelity of Pack. OPBT (blue lines), where $k=aN^\alpha, a=1,\;\alpha\in{(0,1)}$, see Table~\ref{tab:fid_alpha}. $N$~runs through $10^2, 10^3, 10^4$ and $10^5$ as the lines become thicker.}
		\label{fig:asym_f}%
	\end{centering}
\end{figure}

MPBT scheme offers qualitative improvement  over the original scheme for teleporting multipartite states. Even in optimal scheme of standard PBT (i.e. one, with optimal POVMs and the resource state) the number of teleported qubits can only scale as $N^{2/3}$, while in the original scheme, (where the resource state is just $N$ EPR pairs)  the scaling is $N^{1/2}$.

\subsection{Probabilistic protocol} 

Here we shall compare MPBT versus packaged PBT scheme in probabilistic scenario, where the figure of merit is probability of success. Again we will assume that $k=a N^\alpha$. 
We do not have here 
such simple explicit bounds as those given in Theorem \ref{bound_thmF}. 
We shall therefore use exact formula derived in companion paper by means of representation theoretic tools~\cite{Stu2020}. Basing on this result, in Appendix~\ref{AppC} we present exact formula for the probability of success in qubit case expressed in terms  of angular momentum parameters which reads:
  \begin{equation}
  \label{p_succ}
  p_{succ}(N,k)=\frac{1}{2^N}\frac{1}{N+1}\sum_{s=0(\frac{1}{2})}^\frac{N-k}{2}(2s+1)^2\binom{N+1}{\frac{N-k}{2}-s}.
  \end{equation}

\begin{table}[]
    \centering
    \resizebox{8.5cm}{!}{%
    \begin{tabular}{c|c|l|c}  & $p_{s}=0$ & $\qquad \qquad p_{s,\rm cr}$ & $p_s=1$  \\
    \hline
    $Pack. PBT$   & $\alpha>1/3$ & $\alpha_{\rm cr}=1/3\,$, $p_{s,\rm cr}= e^{-ca^{3/2}}$ & $\alpha<1/3$  \\
    $Pack. OPBT$   & $\alpha>1/2$ & $\alpha_{\rm cr}=1/2,\,$ $p_{s,\rm cr}=e^{-3a^2}$ & $\alpha<1/2$  \\
      $MPBT$   & $\alpha>1/2$ & $\alpha_{\rm cr}=1/2, \ $ eq. \eqref{eq26} & $\alpha<1/2$  \\
         $OMPBT$   & $-$ & $\alpha_{\rm cr}=1, \  \ $ $p_{s,\rm cr}=(1+a)^3$ & $\alpha<1$\\
         \hline
    \end{tabular}
    }
	\caption{Table collects the comparison of the asymptotic behaviour of probability of success $p_s$ of all variants of packaged PBT with MPBT in probabilistic version when $k=aN^\alpha$. Here $c=\sqrt{8/\pi}$. By "cr" we denote the critical values of parameter $\alpha$ for which asymptotic value of $p_s$ exhibits a jump.}
    \label{tab:p_alpha}
\end{table}

  Setting $k=a\sqrt{N}$ we can examine the dependence of \eqref{p_succ} on $a$. We can observe, that when $a\to 0$ then $p_{succ} \to 1$. The latter is given (\ref{eq26}), proven in Appendix~\ref{finite} with the means of Central Limit type theorem, although in not completely straightforward way.
  \be
\label{eq26}
\lim_{N \to \infty} p_{succ} = \begin{cases}

2\int_0^\infty x^2\frac{1}{\sqrt{2\pi}}\operatorname{e}^{-\frac{(x+a)^2}{2}}dx,\; & k=a\sqrt{N}, a>0 \\
1, \quad & k=o(\sqrt{N})

\end{cases}
\ee

For PBT,  in the case of original (i.e. nonoptimized) scheme \cite{ishizaka_quantum_2009} we have $p_{succ}\approx 1-c/\sqrt{N}$, where $c=\sqrt{\frac{8}{\pi}}$, and $p_{succ}=1-3/{(3+N)}\approx1-3/{N}$ in the case of OPBT~\cite{ishizaka_quantum_2009}. We may then apply the same reasoning as in the deterministic case. 
Considering (\ref{p_succ}) we can see that 
\begin{equation}
    \lim_{N\rightarrow\infty}p_{succ}^{MPBT}=1 \text{ for } \alpha < 1/2.
\end{equation}
We can thus see that, unlike the deterministic protocol, the regions of asymptotic behaviour of $p_{succ}$ are separated by the same critical value $\alpha_{cr}=1/2$ for packaged OPBT and MPBT probablistic protocols.  Plotting $p_{succ}^{MPBT}$ and $p_{succ}^{OPBT}$ as a function of $a$ in this borderline case, i.e. when $\alpha=\alpha_{cr}=1/2$ (see Figure~\ref{fig:asym}), we observe the new protocol outperforms the previous one in the regime of large $a$. 

\begin{figure}[!]
	\begin{centering}
		\includegraphics[width=0.45\textwidth]{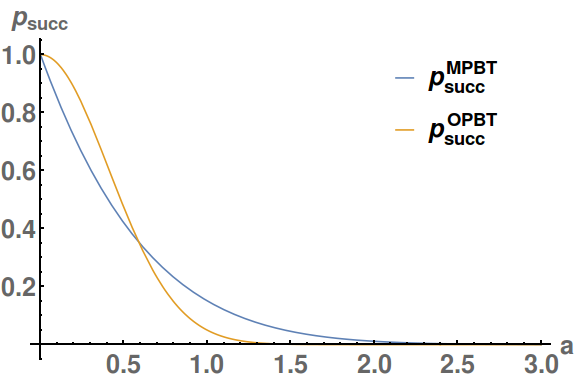}
		\caption{ 
		The limiting probability of success in the probabilistic schemes: Pack. OPBT(orange line) and MPBT (blue line) plotted as a function of $a$ when $k=a\sqrt N$, i.e. $\alpha=1/2$, see Table~\ref{tab:p_alpha}.}
		\label{fig:asym}%
	\end{centering}
\end{figure}
However, one can refer to recent result concerning optimal MPBT~\cite{2020OPT}, where remarkably simple expression for probability of success is evaluated:
\be
p_{succ}^{OMPBT}=\prod\limits_{m=2}^{d^2}\left(1-\frac{k}{N-1+m}\right).
\ee
This for $k=o(N)$ leads to $\lim_{N\rightarrow\infty}p_{succ}^{OMPBT}=1$. We have for $d=2$
\be
\lim_{N\rightarrow\infty}p_{succ}^{OMPBT}=\begin{cases}
(1-a)^3, &k=aN\\
1, &k = o(N)
\end{cases}
\ee

which outperforms the optimal PBT.
This for $k=o(N)$ leads to $\lim_{N\rightarrow\infty}p_{succ}^{OMPBT}=1$. Moreover, for $d=2$ and $k=aN$, that $\lim_{N\rightarrow\infty}p_{succ}^{OMPBT}=(1-a)^3$,
which outperforms the optimal packaged PBT.
We present the comparison of scaling of  the packaged OPBT and OMPBT protocols in the Table~\ref{tab:p_alpha} and Figure~\ref{fig:asym_p}.

\begin{figure}[!]
	\begin{centering}
		\includegraphics[width=0.45\textwidth]{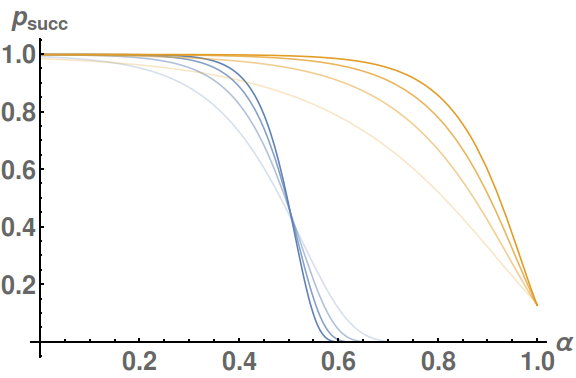}
		\caption{ 
		Probability of success, compared for Pack. OPBT (blue lines), and OMPBT (orange lines) protocols, where $k=aN^\alpha, a=\frac{1}{2}$, see Table~\ref{tab:p_alpha}. 
 $N$~runs through $10^2, 10^3, 10^4$ and $10^5$ as the lines become thicker.}
		\label{fig:asym_p}%
	\end{centering}
\end{figure}
Finally, for a finite number of ports $N$ we have precise lower and upper bound for $p_{succ}$ in our protocol, which are presented in the Appendix~\ref{finite}.

\section{Conclusions and Discussions}
 We have addressed the problem of teleporting a large amount of quantum information. In particular we analyse quantum multi-port teleportation protocol performing transmission of  several quantum systems in one round. 
 By examining improved lower bound on its performance, measured in the entanglement fidelity, we show that the protocol outperforms introduced earlier PBT protocols.
 The bound has been obtained by considering the teleportation process in our protocol as a state discrimination task and it depends only on global parameters like the number of ports, their dimension, and the number of teleported particles.  In particular,  we derived closed expressions for the entanglement fidelity and the probability of success in the described scheme for qubits in the picture of the angular momentum and using Gaussian approximation.
 
 Further, we have shown that in general the number of systems to be teleported can be changed dynamically by the sender with the growing number of ports, still ensuring high efficiency in deterministic and probabilistic scheme.   Even in the optimal PBT scheme, in deterministic case,   
 fidelity can approach $1$ for teleporting  of up to $N^{2/3}$ particles, for our MPBT  $N^\alpha$  particles can be teleported, if only $\alpha<1$.  In probabilistic case, although the rates are equal and the teleportation achieving asymptotically unit probability of success is possible for number of teleported particles of the order $N^\alpha, \alpha < 1/2$, both in MPBT and OPBT, we present result concerning optimal MPBT, which outperforms previous protocols~\cite{2020OPT}.

 We have thus showed that analysed protocol, while still requiring quite mild correction on Bob's site  (just permuting his systems), exhibits qualitatively better "capacity" of transmission.  Our results pave a novel way in teleportation and quantum communication in general. So far in quantum communication, while sending quantum information, one was interested in 
linear rate, i.e. the number of sent qubits per number of channel (or entangled pair) uses. 
On the other hand in quantum computing, it is very important to have teleportation with less correction 
on receiver's side. The port-based teleportation is protocol with virtually no correction. Here,  we consider \textit{ for the first time} the issue of the amount of quantum information sent via 
such teleportation protocols. Our work implies, that within the realm of weakened correction - the proper notion of communication efficiency is determined by   asymptotic exponents in first place.

 
Still, there are a few open questions. First, except probabilistic OMPBT, we do not have exact asymptotic expressions for $d>2$. However, application of central limit theorem or adaptation techniques from~\cite{christ2018asymptotic} theorem should lead to the solution. In the case of deterministic OMPBT the problem of getting the asymptotic behaviour is more complex, since we do not have a closed formula for the fidelity. The fidelity there is given in terms of maximal eigenvalue of teleportation matrix~\cite{StuNJP,Stu2020}, for which we do not have an analytical formula yet, even for $k=1$. Finally, we leave the rigorous analysis of the recycling protocol introduced in~\cite{strelchuk_generalized_2013} and its comparsion with our results for further project.

{\bf Acknowledgements}  MS, MM are supported through grant Sonatina 2, UMO-2018/28/C/ST2/00004 from the Polish National Science Centre. 
Moreover, MH and MM thank the Foundation for Polish
Science through IRAP project co-financed by the EU within
the Smart Growth Operational Programme (contract no.
2018/MAB/5). 
MH also acknowledges support from the National Science Centre, Poland, through grant OPUS 9, 2015/17/B/ST2/01945.
MM and PK would like to thank ICTQT Centre (University of Gda{\'n}sk) for hospitality where part of this work has been done.
\onecolumngrid
\appendix
\section{Symmetries in Multi-port teleportation scheme}
\label{secI}
Before we present argumentation leading us to the proof of Lemma 1 from the main text, we introduce here concepts of permutation operator and its partial transposition. Most of the formalism presented here can be also find in~\cite{Stu2020,Moz1}. Let us consider a representation $V$ of the permutation  group $S(n)$ in the space $\mathcal{H\equiv (\mathbb{C}}^{d})^{\otimes n}$, defined in the following way
	\begin{equation}
	\label{repV}
	\forall \pi \in S(n)\qquad V(\pi ).|e_{i_{1}}\>\otimes |e_{i_{2}}\>\otimes
	\cdots \otimes |e_{i_{n}}\>:=|e_{i_{\pi ^{-1}(1)}}\>\otimes |e_{i_{\pi
			^{-1}(2)}}\>\otimes \cdots \otimes |e_{i_{\pi ^{-1}(n)}}\>,
	\end{equation}
where the set $\{|e_{i}\>\}_{i=1}^{d}$ is an orthonormal basis of the space $\mathcal{\mathbb{C}}^{d}$, and $d$ stands for the dimension. We drop here the lower index in every $i$, since it labels only position of the basis in tensor product $(\mathbb{C}^{d})^{\otimes n}$. This representation is in fact a unitary matrix representation, since it is given with respect to prescribed basis $\{|e_{i}\>\}_{i=1}^{d}$ in the space $\mathcal{\mathbb{C}}^{d}$, and for every $\pi \in S(n)$ the corresponding matrix representation $V(\pi)$ is called a permutation operator. The representation $V$ of $S(n)$ extends in a natural way to the representation of the group algebra $\mathbb{C}\lbrack S(n)]$ and in this way we get the algebra of permutation operators
\be
\mathcal{A}_{n}(d)\equiv \operatorname{span}_{\mathbb{C}}\{V(\pi):\pi \in S(n)\}.
\ee
Having the above definition of the algebra $\mathcal{A}_{n}(d)$, we can introduce a new complex algebra - algebra of the partially transposed permutation operators:
\be
\label{atdef}
\mathcal{A}_{n}^{(k)}(d)\equiv \operatorname{span}_{\mathbb{C}}\{V^{(k)}(\pi ):\pi \in S(n)\}, 
\ee
where the symbol $(k)$ denotes partial transpose operation  with respect to last $k$ systems in the space $\mathcal{(\mathbb{C}}^{d})^{\otimes n}.$   Now, we show that the operator $\rho$ defined in expression (6) from the main text is an element of the algebra $\mathcal{A}_{n}^{(k)}(d)$. Let us denote by $n=N+k$ the number of all systems involved in the teleportation process on Alice's (Bob's) side and introduce the following mapping of indices
\be
\forall \ \mathbf{i}\in\mathcal{I} \quad B_{\mathbf{i}}\mapsto B=\underbrace{B_{n-k+1}B_{n-k+2}\cdots B_{n}}_k.
\ee
Having that we can work with the equivalent form of the operator $\rho$, which is given through the following expression
\be
\rho=\sum_{\mathbf{i}\in \mathcal{I}}\sigma_i=\frac{1}{d^{N-k}}\sum_{\mathbf{i}\in \mathcal{I}}\mathbf{1}_{\bar{A}_{\mathbf{i}}}\ot P^+_{A_{\mathbf{i}}B},
\ee
where for a given index $\mathbf{i}=\{i_1,i_2,\ldots,i_k\}$ the operator $\mathbf{1}_{\bar{A}_{\mathbf{i}}}\ot P^+_{A_{\mathbf{i}}B}$ is of the following form
\be
\begin{split}
\sigma_{\mathbf{i}}&=\mathbf{1}_{\bar{A}_{\mathbf{i}}}\ot P^+_{A_{\mathbf{i}}B}=\mathbf{1}_{\bar{A}_{\mathbf{i}}}\ot P^+_{i_1,n-k+1}\ot P^+_{i_2,n-k+2}\ot \cdots \ot P^{+}_{i_k,n}\\
&=\frac{1}{d^N}\mathbf{1}_{\bar{A}_{\mathbf{i}}}\ot V^{t_{n-k+1}}[(i_1,n-k+1)]\ot V^{t_{n-k+2}}[(i_2,n-k+2)]\ot \cdots \ot V^{t_{n}}[(i_k,n)].
\end{split}
\ee
By $t_{n-k+1},t_{n-k+2},\ldots,t_n$ we denote transposition with respect to systems on positions $n-k+1,n-k+2,\ldots,n$ respectively.   Next, we have to rewrite $\rho$ in more convenient form for our further considerations. First, we observe that we can distinguish an index $\mathbf{i}_0=\{n-k,n-k-1,\ldots,n-2k+1\}$, for which the corresponding signal is called canonical:
\be
\begin{split}
\sigma_{\mathbf{i}_0}&=\frac{1}{d^N}\mathbf{1}_{\bar{A}_{\mathbf{i}_0}}\ot V^{t_{n-k+1}}[(n-k,n-k+1)]\ot V^{t_{n-k+2}}[(n-k-1,n-k+2)]\ot \cdots \ot V^{t_{n}}[(n-2k+1,n)]\\
&=\frac{1}{d^N}\left(\mathbf{1}_{\bar{A}_{\mathbf{i}_0}}\ot V[(n-k,n-k+1)]\ot V[(n-k-1,n-k+2)]\ot \cdots \ot V[(n-2k+1,n)]\right)^{t_{n-k+1} t_{n-k+2} \cdots t_n}\\
&=\frac{1}{d^N}V^{(k)},
\end{split}
\ee
by $V$ we denote the total permutation operator in the bracket with this specific composition of permutations, including the identity operator. Having the definition of $\sigma_{\mathbf{i}_0}$, let us observe that any other signal can be obtained from it by acting of $V(\tau)$, where $\tau$ is permutation from the coset $\mathcal{S}_{n,k}=S(n-k)/S(n-2k)$, so
\be
\label{rr}
\rho=\frac{1}{d^N}\sum_{\tau \in \mathcal{S}_{n,k}}V(\tau)V^{(k)}V^{\dagger}(\tau).
\ee
This clearly shows that the operator $\rho$, as well as every operator $\sigma_{\mathbf{i}}$ belongs to the algebra $\mathcal{A}_n^{(k)}(d)$ defined in~\eqref{atdef}.
The cardinality of $\mathcal{S}_{n,k}$ gives us number of all possible signals $k!\binom{N}{k}$. Since all the operators $V(\tau)$ are invariant with respect to the composition of partial transpositions $(k)$, we can rewrite $\rho$ from~\eqref{rr} as in the following definition:
\begin{definition}
		\label{def9}
		For $N$ ports the operator $\rho$ from expression~\eqref{rr} can be re-written in the following form:
	\begin{equation}
	\rho
	\equiv \sum_{a_{1},a_{2},\ldots,a_{k}=1}^{n-k}V^{(k)}[(a_{k},n-k+1)(a_{k-1},n-k+2).\cdots(a_{1},n)],
	\end{equation}%
	where all numbers $a_{1},a_{2},\ldots,a_{k}$ are different and $k\leq \floor{ 
	\frac{N}{2}}$. For compactness of the further calculations  we drop here the normalisation constant $1/d^N$ in form of every signal $\sigma_{\mathbf{i}}$.
\end{definition}

\section{Proof of Lemma 1 and Theorem 1 from the main text}
\label{AppB}
Having discussion on symmetries in multi-port based teleportation protocols we are in position to compute the lower bound on the entanglement fidelity $F$. As it was pointed in the main text, and in~\cite{ishizaka_asymptotic_2008,beigi_konig}, the entanglement fidelity $F$ can be connected with probability of success of the state discrimination $p_{dist}$ of the ensemble $\mathcal{E}$:
\be
\mathcal{E}=\left\{\frac{1}{k!\binom{N}{k}},\sigma_{\mathbf{i}}\right\},
\ee
where $\sigma_{\mathbf{i}}$ are the signal states described in Section~\ref{secI}. It can be shown that $F$ and $p_{dist}$ are related to each other by the following relation
\be
F=\frac{k!\binom{N}{k}}{d^{2k}}p_{dist},\qquad p_{dist}=\frac{1}{k!\binom{N}{k}}\sum_{\mathbf{i}\in\mathcal{I}}\tr(\Pi_{\mathbf{i}}\sigma_{\mathbf{i}})\geq \frac{1}{k!\binom{N}{k}r\tr \overline{\rho}^2},
\ee
where $r=(1/(k!\binom{N}{k}))\sum_{\mathbf{i}\in\mathcal{I}}\operatorname{rank}(\sigma_{\mathbf{i}})=d^{N-k}$, and the lower bound on $p_{dist}$ is a simple generalisation of Lemma A.3 from~\cite{beigi_konig}. This means that $F$ is bounded from the below as
\be
\label{rhs}
F\geq \frac{1}{d^{N+k}\tr(\overline{\rho}^2)},
\ee
where the operator $\overline{\rho}$ is normalised version ($\tr \overline{\rho}=1$) of the operator $\rho$ from Definition~\ref{def9}:
\be
\label{norm0}
\overline{\rho}=\frac{\rho}{\tr(\rho)}=\frac{1}{d^{N}k!\binom{N}{k}}\rho,
\ee
since we can observe that for operator $\rho$ from Definition~\ref{def9} the following holds
\be
\label{trrho}
\begin{split}
\tr(\rho)&=d^{n-k}(n-k)(n-k-1)\cdot \ldots \cdot (n-2k+1)=d^{N}N(N-1)\cdot\ldots \cdot (N-k+1)=d^Nk!\binom{N}{k}\\
&=d^N\frac{N!}{(N-k)!},
\end{split}
\ee
where $N=n-k$ denotes number of ports. It means that normalised operator $\overline{\rho}$ can be written finally as
\be
\label{norm}
\overline{\rho}=\frac{\rho}{\tr(\rho)}=\frac{1}{d^{N}\frac{N!}{(N-k)!}}\rho.
\ee
Later on we use form of $\overline{\rho}$ from expression~\eqref{norm}.

The above considerations show that to compute the right-hand side of~\eqref{rhs}  the crucial is evaluation of $\tr(\rho^2)$.  The calculation of $\operatorname{Tr}(\rho^2)$ is technically complicated and lengthy and it contains a few intermediate steps and auxiliary results which we present separately as Lemma \ref{L4} to \ref{L15}. The final result, i.e. the formula for $\operatorname{Tr}(\rho^2)$ is given in Theorem~\ref{thm8} . We will use in the derivation of this  result the following useful notation
\begin{notation}
\label{notation}
	We define inductively a sequence of transpositions%
	\be
	L_{1}=(a_{1},b_{1}), \ L_{2}=(a_{2},L_{1}(b_{2}\
	)), \ L_{3}=(a_{3},L_{2}L_{1}(b_{3}\ )),\ldots, L_{k}=(a_{k},L_{k-1}\cdots L_{2}L_{1}(b_{k}\ )),
	\ee
	\be
	L_{1}^{\prime }=(b_{1},a_{1}), \ L_{2}^{\prime
	}=(b_{2},L_{1}^{\prime }(a_{2}\ )), \ L_{3}^{\prime
	}=(b_{3},L_{2}^{\prime }L_{1}^{\prime }(a_{3}\ )),\ldots,	L_{k}^{\prime }=(b_{k},L_{k-1}^{\prime }\cdots L_{2}^{\prime }L_{1}^{\prime
}(a_{k}\ )),
	\ee
	where $(a,b)$ is simply a transposition acting on some $c$ as

	\be
	(a,b)(c)=\begin{cases}
		c, & \text{when }a,b \neq c\\
		b, & \text{when }a=c\\
		a, & \text{when }b=c
	\end{cases}
		\ee
	and where the numbers $a_{1},a_{2},\ldots,a_{k}$ are different and similarly the numbers 
	$b_{1},b_{2},\ldots,b_{k}$ are different. For such a transpositions we define
	reversed  transpositions also defined on the same numbers $%
	a_{1},a_{2},\ldots,a_{k}$ and $b_{1},b_{2},\ldots,b_{k}$ in the following way so in the primed transpositions the numbers $a_{1},a_{2},\ldots,a_{k}$ and $b_{1},b_{2},\ldots,b_{k}$ are interchanged.
\end{notation}
The above transpositions satisfy the relation which
easily follows from the structure of the transpositions
\begin{lemma}
	\label{L4}
	We have 
	\begin{equation}
	L_{s}^{\prime }L_{s-1}^{\prime }\cdots L_{1}^{\prime }=L_{1}L_{2}\cdots L_{s},\qquad
	s=1,\ldots,k.
	\end{equation}
\end{lemma}
Using introduced notation we may derive the following composition rule
\begin{lemma}
	\label{P12}
	Let the numbers $1\leq a_{1},a_{2},\ldots,a_{k}\leq n-k$ are different and
	similarly numbers $1\leq b_{1},b_{2},\ldots,b_{k}\leq n-k$ are different and $%
	k\leq \floor{\frac{n}{2}}$, then 
	\be
	\begin{split}
	&V^{(k)}[(a_{k},n-k+1)(a_{k-1},n-k+2)\cdots (a_{1},n)]V^{(k)}[(b_{k},n-k+1)(b_{k-1},n-k+2)\cdots (b_{1},n)]\\
	&=d^{\delta _{a_{k},L_{k-1}\cdots L_{2}L_{1}(b_{k})}}d^{\delta
		_{a_{k-1},L_{k-2}\cdots L_{2}L_{1}(b_{k-1})}}\cdots d^{\delta
		_{a_{2,L_{1}(b_{2})}}}d^{\delta
		_{a_{1,}b_{1}}}V(L_{k})V(L_{k-1})\cdots V(L_{1})\times\\
	&\times V^{(k)}[(b_{k},n-k+1)(b_{k-1},n-k+2)\cdots (b_{1},n)]. 
	\end{split}
	\ee
\end{lemma}
Applying the statement of Proposition~\ref{P12} twice, together with Lemma~\ref{L4}, we get
\begin{lemma}
	\label{lem8}
	Let the numbers $1\leq a_{1},a_{2},\ldots,a_{k}\leq n-k$ are different and
	similarly the numbers $1\leq b_{1},b_{2},\ldots,b_{k}\leq n-k$ are different and $k\leq \floor{\frac{n}{2}}$, then 
	\be
	\label{exp}
	\begin{split}
	&V^{(k)}[(a_{k},n-k+1)(a_{k-1},n-k+2)\cdots(a_{1},n)]V^{(k)}[(b_{k},n-k+1)(b_{k-1},n-k+2)\cdots(b_{1},n)]\times\\
	&\times V^{(k)}[(a_{k},n-k+1)(a_{k-1},n-k+2)\cdots(a_{1},n)]\\
&=d^{2\delta _{a_{k},L_{k-1}\cdots L_{2}L_{1}(b_{k})}}d^{2\delta
		_{a_{k-1},L_{k-2}\cdots L_{2}L_{1}(b_{k-1})}}\cdots d^{2\delta
		_{a_{2,L_{1}(b_{2})}}}d^{2\delta _{a_{1,}b_{1}}}\times\\ 
	&\times V^{(k)}[(a_{k},n-k+1)(a_{k-1},n-k+2)\cdots (a_{1},n)]. 
	\end{split}
	\ee
\end{lemma}
From expression~\eqref{exp} we get immediately
\begin{corollary}
\label{cor9}
	Let the numbers $1\leq a_{1},a_{2},\ldots,a_{k}\leq n-k$ are different and
	similarly numbers $1\leq b_{1},b_{2},\ldots,b_{k}\leq n-k$ are different and $k\leq \floor{\frac{n}{2}}$, then 
	\be
	\label{formula}
	\begin{split}
	&\tr\left[
	V^{(k)}[(a_{k},n-k+1)(a_{k-1},n-k+2)\cdots (a_{1},n)]V^{(k)}[(b_{k},n-k+1)(b_{k-1},n-k+2)\cdots (b_{1},n)]%
	\right]\\
	&=d^{n-2k}d^{2\delta _{a_{k},L_{k-1}\cdots L_{2}L_{1}(b_{k})}}d^{2\delta
		_{a_{k-1},L_{k-2}\cdots L_{2}L_{1}(b_{k-1})}}\cdots d^{2\delta
		_{a_{2,L_{1}(b_{2})}}}d^{2\delta _{a_{1,}b_{1}}}.
	\end{split}
\ee
\end{corollary}
Using equation~\eqref{formula} we get that 
\begin{equation}
\tr(\rho
^{2})=\sum_{a_{i}}%
\sum_{b_{i}}\tr[V^{(k)}[(a_{k},n-k+1)(a_{k-1},n-k+2)\cdots (a_{1},n)]V^{(k)}[(b_{k},n-k+1)(b_{k-1},n-k+2)\cdots (b_{1},n)],
\end{equation}%
where the summation is all over pairwise different $1\leq
a_{1},a_{2},\ldots,a_{k}\leq n-k$ and $1\leq b_{1},b_{2},\ldots,b_{k}\leq n-k,$ is equal to%
\begin{equation}
\tr(\rho ^{2})=\sum_{a_{i}}\sum_{b_{i}}d^{n-2k}d^{2\delta
	_{a_{k},L_{k-1}\cdots L_{2}L_{1}(b_{k})}}d^{2\delta
	_{a_{k-1},L_{k-2}\cdots L_{2}L_{1}(b_{k-1})}}\cdots d^{2\delta
	_{a_{2,L_{1}(b_{2})}}}d^{2\delta _{a_{1,}b_{1}}}.
\end{equation}%
In order to calculate this sum we need the following

\begin{lemma}
	\label{L15}
Using Notation~\ref{notation} we have the following property:
	\begin{equation}
	b_{s}\neq b_{1},b_{2},\ldots,b_{s-1}\Rightarrow
	L_{s-1}L_{s-2}\cdots L_{1}(b_{s})\neq a_{1},a_{2},\ldots,a_{k-1},\quad s=1,\ldots,k.
	\end{equation}
\end{lemma}

In the particular case, for $k=2$, the statement of the above lemma reduces to
\begin{equation}
b_{2}\neq b_{1}\Rightarrow L_1(b_{2})\neq a_{1}.
\end{equation}%
Applying  Lemma~\ref{L15} to each sum appearing on the $RHS$ of equation for $%
\tr(\rho ^{2})$ we are in position to formulate the main result of this section:

\begin{theorem}
\label{thm8}
	Let $\rho $ be as in Definition~\ref{def9}, then 
	\begin{equation}
	\tr(\rho
	^{2})=d^{n-2k}(n-2k-1)(n-2k-2)\cdots (n-k)(d^{2}+n-2k)(d^{2}+n-2k+1)\cdots (d^{2}+n-k-1),
	\end{equation}
	or equivalently, in a more compact form
	\be
	\label{square2}
	\tr(\rho
	^{2})=d^{N-k}\frac{N!}{(N-k)!}\frac{(d^2+N-1)!}{(d^2+N-k-1)!},
	\ee
	where $N=n-k$ is number of port in MPBT protocol.
\end{theorem}

\begin{proof}
	In order to better explain the role of Lemma~\ref{L15} in the proof, let us consider first the case of $k=2$, where we have 
	\begin{equation}
	\tr(\rho ^{2})=d^{n-4}\sum_{b_{i}}\left( \sum_{a_{i}}d^{2\delta _{a_{2},\tau
			_{(a_{1},b_{1})}(b_{2})}}d^{2\delta
		_{a_{1,}b_{1}}}\right) =d^{n-4}\sum_{b_{i}}\left( \sum_{a_{1}}\sum_{a_{2}\neq
		a_{1}}d^{2\delta _{a_{2},\tau _{(a_{1},b_{1})}(b_{2})}}d^{2\delta
		_{a_{1,}b_{1}}}\right) ,
	\end{equation}%
	where $1\leq a_{1},a_{2}\leq n-2$ and $1\leq b_{1},b_{2},\leq n-2$ are
	pairwise different. Now using Lemma~\ref{L15} we get 
	\begin{equation}
	\sum_{a_{2}\neq a_{1}}d^{2\delta _{a_{2},\tau
			_{(a_{1},b_{1})}(b_{2})}}=d^{2}+n-4,
	\end{equation}%
	since $a_{2}$ runs over the set $%
	\{1,2,\ldots,n-2\}\backslash \{a_{1}\}$ exactly once is equal to the number $%
	\tau _{(a_{1},b_{1})}(b_{2})\neq a_{1}$ and $n-4$ times is not equal to $%
	\tau _{(a_{1},b_{1})}(b_{2}).$ This implies
	\begin{equation}
	\begin{split}
	\tr(\rho ^{2})&=d^{n-4}(d^{2}+n-4)\sum_{b_{1}}\sum_{b_{2}\neq
		b_{1}}\left( \sum_{a_{1}}d^{2\delta
		_{a_{1,}b_{1}}}\right) =d^{n-4}(n-3)(d^{2}+n-4)\sum_{b_{1}=1}^{n-2}%
	\sum_{a_{1}=1}^{n-2}d^{2\delta _{a_{1,}b_{1}}}\\
	&=d^{n-4}(n-2)(n-3)(d^{2}+n-3)(d^{2}+n-4).
	\end{split} 
	\end{equation}%
	For general case the way of proving is the same but we have more steps.
\end{proof}
We illustrate the statement of
Lemma~\ref{thm8} by two following examples:
\begin{example}
	Let $k=1$, then 
	\begin{equation}
	\tr(\rho ^{2})=d^{n-2}(n-1)(d^{2}+n-2).
	\end{equation}%
	This expression can be also evaluated using group-theoretic approach, see the proof of Theorem 52 on page 28 in~\cite{MozJPA}.
\end{example}

\begin{example}
	Let $k=2$, then 
	\begin{equation}
	\tr(\rho ^{2})=d^{n-4}(n-2)(n-3)(d^{2}+n-3)(d^{2}+n-4).
	\end{equation}
Here, the group-theoretic approach would demand new summation rules for the irreducible representations, which are yet not know.
\end{example}

All the above described considerations are in fact proof of Lemma 1 from the main text. Indeed, taking expression~\eqref{square2} and~\eqref{trrho} we write the following for the normalised operator $\overline{\rho}$ from~\eqref{norm}:
\begin{equation}
\begin{split}
\tr(\overline{\rho}^2)&=\frac{\tr(\rho^2)}{(\tr(\rho))^2}
=\frac{d^{N-k}}{d^{2N}}\frac{\frac{N!}{(N-k)!}\frac{(d^2+N-1)!}{(d^2+N-k-1)!}}{\left(\frac{N!}{(N-k)!}\right)^2}\\
&=d^{-N-k}\left[\frac{N!}{ (N-k)!}\right]^{-1} \frac{ (d^2 + N - 1)!}{(d^2 + N - k - 1)!}\\
&=d^{-N-k}\binom{N}{k}^{-1}\binom{d^2+N-1}{k}.
\end{split}
\end{equation}

Now, using relation $n=N+k$ and expression~\eqref{trrho} we rewrite~\eqref{rhs} as
\be
\label{1a}
\begin{split}
F\geq \frac{1}{d^{N+k}\tr(\overline{\rho}^2)}=\binom{N}{k}\binom{d^2+N-1}{k}^{-1}.
\end{split}
\ee
which recovers the first bound (9) from  Theorem 2 in the main text. To prove the second bound in (9) from the same theorem let us observe the following:
\be
\label{1b}
\begin{split}
\binom{N}{k}\binom{d^{2}+N-1%
}{k}^{-1}&= \prod\limits_{s=0}^{k-1}\left(1-\frac{d^2-1}{d^2+N-1 - s}\right)\geq \left(1-\frac{d^2-1}{d^2 + N - k}\right)^k.
\end{split}
\ee

Now, let us prove that the entanglement fidelity $F$ scales at least as $1-O(1/N)$ and goes to 1 with $N\rightarrow \infty$. To see this is is enough to apply the  Bernoulli inequality $(1-x)^k\geq 1-k x$ (valid for $x<1$) to $\left(1-\frac{d^2-1}{d^2-N-k}\right)^k$.

There is also an alternative proof of the fidelity scaling, showing interesting connection with the symmetric polynomials. To show it, first we have to  derive a different  lower bound for the entanglement fidelity exploiting symmetric polynomials. Defining variables $x_s\equiv 1/(d^2+N-s-1)$ for $s=0,\ldots,k-1$, and completely symmetric polynomials $S_l(x_1,x_2,\ldots,x_k)$ in variables $x_1,x_2,\ldots,x_k$, for $l\in\mathbb{N}$. For example we have:
\be
S_0(x_1,x_2,\ldots,x_k)=1,\quad S_1(x_1,x_2,\ldots,x_k)=x_1+x_2+\ldots +x_k,\quad S_2(x_1,x_2,\ldots,x_k)=\sum_{1\leq j\leq l\leq k}x_jx_l\quad \text{etc.}
\ee
Using the above, we can write the bound~\eqref{1a} with~\eqref{1b} as
\be
F\geq  \prod\limits_{s=0}^{k-1}\left(1-\frac{d^2-1}{d^2+N-s-1}\right)=\sum_{l=0}^k (-1)^l (d^2-1)^lS_l(x_1,x_2,\ldots,x_k).
\ee
To evaluate the leading term in obtained bounds on the entanglement fidelity we have to use observation that we can  write the right-hand side of (9) in terms of symmetric polynomials $S_l(x_1,\ldots,x_k)$ with $x_s=1/(d^2+N-s-1)$ for $s=0,\ldots,k-1$:
\be
\label{bound2a1}
\begin{split}
F&\geq  \sum_{l=0}^k (-1)^l (d^2-1)^lS_l(x_1,x_2,\ldots,x_k)=1-\sum_{s=0}^{k-1} \frac{d^2-1}{d^2+N-s-1}+O(1/N^2).
\end{split}
\ee
We then see from \eqref{bound2a1} that when $k$ is fixed and $N$ goes to infinity, the right hand side  approaches $1$. 

In the end of this appendix we briefly explain the flow of the reasoning presented~\cite{strelchuk_generalized_2013}. The idea presented there is also based on the state discrimination problem and computing $\tr \overline{\rho}^2$. However, argument for computing trace from the composition $\sigma_{\mathbf{i}}\sigma_{\mathbf{j}}$, for maximally entangled pairs fully overlapping, does not take into account that the resulting operators are  in general permutation operators with different number of disjoint cycles. The final trace depends on the number of disjoint cycles and cannot be characterised only by global parameters, but has to take into account the interior structure of permutations. We illustrate this in the example below.
\begin{example}
\label{ex:contr}
Let us consider two signals for $n=6, k=2$, and arbitrary $d$ of the form
\be
\sigma_1=\frac{1}{d^4}V^{(2)}[(45)(36)],\qquad  \sigma_2=\frac{1}{d^4}V^{(2)}[(35)(46)].
\ee
Trace from their overlap, according to Corollary~\eqref{cor9}, equals to
\be
\begin{split}
\tr\left(\sigma_1\sigma_2\right)=\frac{1}{d^8}\tr\left(V^{(2)}[(45)(36)]V^{(2)}[(35)(46)]\right)=\frac{d^2}{d^8}d^{2\delta_{4,L_1(3)}}d^{\delta_{3,4}}=\frac{d^2}{d^8}d^{2\delta_{4,4}}=\frac{1}{d^4},
\end{split}
\ee
since $L_1(3)=(3,4)[3]=4$ (see Notation~\ref{notation}). However, in~\cite{strelchuk_generalized_2013} authors argue that the value of $\tr\left(\sigma_1\sigma_2\right)$ equals to $1/d^6$, or in general $1/d^{N+k}$, which is in contradiction with mentioned Corollary~\ref{cor9}.
Note that this particular example can be evaluated directly, without referring to the corollary, by exploiting, in particular, the simple fact, that 
$\tr A^\Gamma B^\Gamma=\tr AB$ 
for operators $A,B$, where $\Gamma$ is partial transpose, and 
the fact that trace of operator of permutation is equal to $d^c$ where $c$ is number of cycles of the permutation, including trivial cycles. 
\end{example}
Fortunately, some part of argumentation from~\cite{strelchuk_generalized_2013} can be saved. Namely, quantity $1/d^{N+k}$ always lower bounds $\tr(\sigma_{\mathbf{i}}\sigma_{\mathbf{j}})$ for signals with overlapping maximally entangled states. This means that the result in~\cite{strelchuk_generalized_2013}  actually gives lower bound on $\tr \overline{\rho}^2$ and then on entanglement fidelity. Here, we do not compute bound on $\tr \overline{\rho}^2$, but we evaluate it explicitly. This leads us to more effective and elegant lower bound on efficiency of the deterministic MPBT.

\section{Derivation of the entanglement fidelity and probability of success for qubits}
\label{AppC}
Before we present here the main consideration leading us to closed expressions for the entanglement fidelity and average probability of success, we introduce basic ideas concerning representation theory of permutation group $S(n)$. For more details we refer reader to~\cite{Fulton1991-book-rep,FultonTab}

Every partition (shape) of natural number $n$ is a sequence $\alpha=(\alpha_1,\ldots,\alpha_r)$ satisfying
\be
\label{partition}
\forall_{i} \  \alpha_i \geq 0,\quad \alpha_1\geq \alpha_2\geq \ldots \geq \alpha_r,\quad \sum_{i=i}^r\alpha_i=n,
\ee
where $r\in \{1,\ldots,n\}$. Every such partition corresponds with some frame, which is called {\bf Young frame}. The Young frame associated with partition $\mu$ is the array formed by $n$ boxes with $l$ left-justified rows.
\begin{example}
Let us take $n=4$. Allowed partitions are $\alpha_1=(4), \alpha_2=(3,1), \alpha_3=(2,2), \alpha_4=(2,1,1), \alpha_5=(1,1,1,1)$ and they correspond to the following Young frames:
\be
\label{ex}
\yng(4) \quad \yng(3,1) \quad \yng(2,2) \quad \yng(2,1,1) \quad \yng(1,1,1,1)
\ee
\end{example}
Any irreducible representation (irrep) of symmetric group $S(n)$ is labelled by Young frames with $n$ boxes. From above example we see that for $n=4$ we have five possible irreps - because we have only five different Young frames. 

A {\bf Standard Young Tableaux (SYT)} is a Young frame of shape $\alpha$ of $n$ objects, such that labels occur increasing in every row from left to the right, and increasing in every column from top to the downwards. One can see that every number between 1 and $n$ occurs only once.
\begin{example}
For the case $n=4$ and $\alpha=(2,1,1)$ we have three possible fillings
\be
\label{c1}
\young(12,3,4)\qquad \young(13,2,4) \qquad \young(14,2,3)
\ee
\end{example}
The number of STYs for given shape $\alpha$ and natural number $n$ can be obtained by combinatorial rules.  It turns out that the dimension $d_{\alpha}$ of an irreducible representation of symmetric group $S(n)$ for fixed $\alpha$ is equal to the number of SYTs.

A {\bf semi-standard Young Tableaux (SSYT)} is a Young frame of shape $\alpha$ of $n$ boxes filled with numbers $\{1,2,\ldots,s\}$ for some $s$, such that rows are weakly and columns are strictly increasing.
\begin{example}
For the case $n=4$, $\alpha=(3,1)$, and set $\{1,2\}$, we have three possible fillings
\be
\label{c1}
\young(111,2)\qquad \young(112,2) \qquad \young(122,2).
\ee
\end{example}
The total number of SSYTs also can be computed by combinatorial rules, and this number equals to the multiplicity $m_{\alpha}$ in the space $(\mathbb{C}^d)^{\ot n}$ by plugging $s=d$. One can notice, that whenever $d$ is smaller than the length of the longest column in $\alpha$ the corresponding multiplicity equals zero. This means that such irrep does not occur into decomposition. For example considering the qubit case, when $d=2$, only three first Young frames from~\eqref{ex} appear. More formally, considering representations of $S(n)$ on space $(\mathbb{C}^d)^{\ot n}$, where $d$ denotes dimension, the irreducible representations of $S(n)$ are labelled by Young frames with $n$ boxes, restricted to the ones with at most $d$ rows. We denote this fact by $\alpha \vdash_d n$ (in short $\alpha \vdash n$ when it is clear from the context). 

Additionally, the symbol  $\mu \in \alpha$ denotes a Young frames $\mu \vdash n$ obtained from a Young frame $\alpha \vdash n-k$ by adding $k$ boxes one by one, getting in every step a valid Young frame. While by the symbol $\alpha \in \mu$ we denote Young frames $\alpha \vdash n-k$ obtained from a Young frame $\mu \vdash n$ by subtracting $k$ boxes. Such a procedure can be performed on many different ways, and total number of them we denote as $m_{\mu/\alpha}$. In the general situation we do not have closed expression for computing $m_{\mu/\alpha}$, however in the qubit case (Young frames up to two rows) we have \cite{aitken_1943,xxx,adin2014enumeration}
\be
\label{num1}
m_{\mu/\alpha} = k!\operatorname{det}\left(\frac{1}{\alpha_i - \mu_j  - i + j}\right)_{i,j=1,2}.
\ee

Having the above description of all the necessary objects and facts, we are in position to prove expressions for the entanglement fidelity $F$ and averaged probability of success $p_{succ}$ form the main text.
In general, for qudits, the formulas are expressed in terms of: 
\begin{enumerate}[i)]
\item  dimensionality $d_{\mu}$ of the irrep $\mu$
\item its multiplicity $m_{\mu}$ in the representation of $S(n)$ permuting $n$ qdits.
\item numbers $m_{\mu/\alpha}$ describing number of possibilities of obtaining a frame $\mu \vdash n$ from a frame $\alpha \vdash n-k$ by adding $k$ boxes, 
\end{enumerate}
and have the following form for the entanglement fidelity~\cite{Stu2020}:
\be
	\label{eq:fidelity}
	F=\frac{1}{d^{N+2k}}\sum_{\alpha \vdash N-k}\left(\sum_{\mu\in\alpha}m_{\mu/\alpha} \sqrt{m_{\mu}d_{\mu}}\right)^2.
	\ee
	and the averaged probability of success~\cite{Stu2020}:
	\be
	\label{eq:p_succ}
	p_{succ}=\frac{1}{d^N}\sum_{\alpha \vdash N-k}m_{\alpha}^2\mathop{\operatorname{min}}\limits_{\mu\in\alpha}\frac{d_{\mu}}{m_{\mu}}.
	\ee

For qubits we shall see that these quantities are expressed in more familiar language of angular momentum.
Having a Young frame $\alpha \vdash N-k$  we denote the Young frame $\mu \vdash N$ obtained from $\alpha$ by adding consecutive $k$ boxes 
 by $\mu\in\alpha$. The number of such allowed additions is $m_{\mu/\alpha}$. 
The multiplicity $m_{\mu}$ equals dimension of spin-$j$ representation of $SU(2)^{\otimes(N-k)}$, where $2j=\mu_1-\mu_2$ and thus $m_\mu=2j+1$; the dimension $d_{\mu}$ corresponds to multiplicity of spin-$j$ representation, given by $\frac{(2j+1)(N-k)!}{((N-k)/2 + j +1)!((N-k)/2-j)!}$. 
Finally, expression~\eqref{num1} can be expressed as 
\be
m_{2s,2j,k}={ k \choose s -j+k/2}-{k \choose  s + j + k/2 + 1}.
\ee
Thus, considering that spin of $N-k$ spin-$\frac{1}{2}$ particles can take values from 0 or $1/2$ depending on the parity of $N-k$ to $\frac{N-k}{2}$ and after adding $k$ particles either from $0(\frac{1}{2})$ or $s-k/2$, whichever quantity is greater, to $s+k/2$,
We can rewrite  expression (\ref{eq:fidelity}) as
	\be
	\label{F_qubits}
	F = \frac{1}{2^{N+2k}}\sum_{s=0(\frac{1}{2})}^\frac{N-k}{2}\left( \sum_{j=\operatorname{max}\{0(\frac{1}{2}), s-\frac{k}{2}\}}^{s+\frac{k}{2}} \left({ k \choose s -j+k/2}-{k \choose s + j + k/2 + 1}\right)\sqrt{\frac{(2j+1)^2N!}{(N/2-j)!(N/2+j+1)!}}\right)^2.
\ee
Defining auxiliary quantity $h_{sjk}$ of the form
\begin{align}
    h_{sjk}=
    { k \choose s -j+k/2}-{k \choose s + j + k/2 + 1},
\end{align}
and simplifying expression under the square root, we arrive to the final formula for the fidelity:
\begin{align}
    \label{F_qubits}
    F = \frac{1}{2^{N+2k}}\sum_{s=0(\frac{1}{2})}^\frac{N-k}{2}\left( \sum_{j=\operatorname{max}\{0(\frac{1}{2}), s-\frac{k}{2}\}}^{s+\frac{k}{2}}
    \frac{2j+1}{N+1}
    \sqrt{{N+1 \choose \frac{N}{2}-j}}\,  h_{sjk}
    \right)^2
\end{align}
On Figure~\ref{FPBT21b} we show comparison of the first bound (9) from Theorem 2 of the main text (Appendix~\ref{AppB}) with the exact expression from expression~\eqref{F_qubits}.
\begin{figure}[!]
	\begin{centering}
		\includegraphics[width=0.5\textwidth]{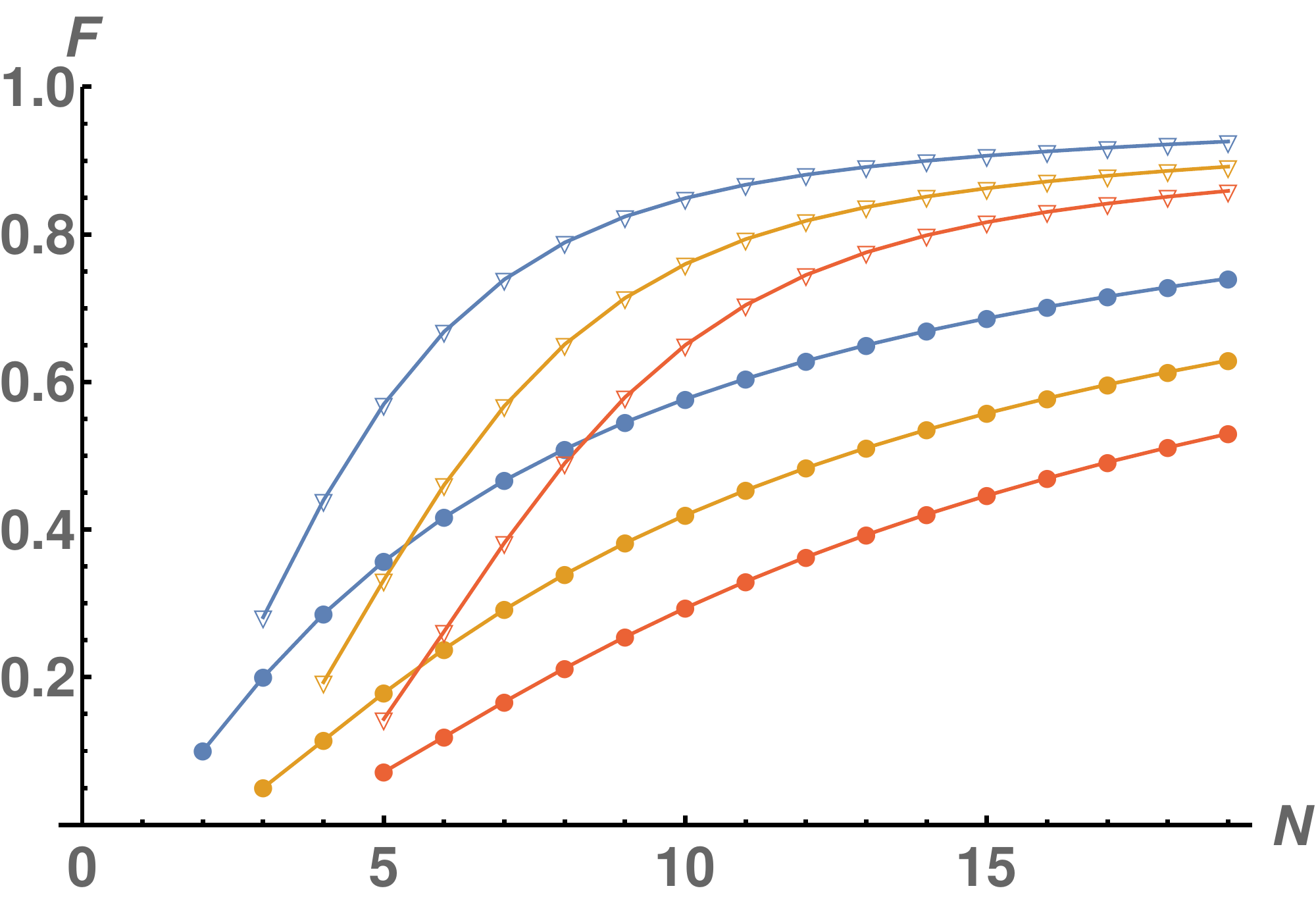}
		\caption{Comparison of the first bound (9) on entanglement fidelity from Theorem 2 in the main text (full circles (full circles, for $k=2$ marked with \protect\raisebox{-.5mm}{\protect\includegraphics[height=3mm]{circ_blue.png}}, $k=3$ \protect\raisebox{-.5mm}{\protect\includegraphics[height=3mm]{circ_orange.png}}, $k=4$		\protect\raisebox{-.5mm}{\protect\includegraphics[height=3mm]{circ_red.png}})) with exact values for the qubit case (reversed triangle markers, for $k=2$ marked with \protect\raisebox{-.5mm}{\protect\includegraphics[height=4mm]{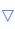}}, $k=3$ \protect\raisebox{-.5mm}{\protect\includegraphics[height=4mm]{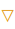}}, $k=4$ \protect\raisebox{-.5mm}{\protect\includegraphics[height=4mm]{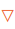}}).
} 
		\label{FPBT21b}%
	\end{centering}
\end{figure}

In case of probability of success, using directly the argumentation from Section V, we rewrite expression \eqref{eq:p_succ} as:
\be
	p_{succ}(N, k)= \frac{1}{2^{N}}\sum_{s=0(\frac{1}{2})}^\frac{N-k}{2} \min_{j} \frac{N!(2s+1)^2 }{(N/2-j)!(N/2+j+1)!}=\frac{1}{2^{N}}\sum_{s=0(\frac{1}{2})}^\frac{N-k}{2}\frac{N!(2s+1)^2}{(N/2-s-k/2)!(N/2+s+k/2+1)!}.
\ee
Observing that
\be
\frac{N!}{(N/2-s-k/2)!(N/2+s+k/2+1)!}=\frac{1}{N+1}\binom{N+1}{\frac{N-k}{2}-s}.
\ee
we get the final form given in by the following equation 
\be
\label{exp50}
p_{succ}(N, k)=\frac{1}{2^{N}}\sum_{s=0(\frac{1}{2})}^\frac{N-k}{2}\frac{1}{N+1}\binom{N+1}{\frac{N-k}{2}-s}.
\ee
 In Section~\ref{finite} of this appendix we discuss additional properties of probability of success $p_{succ}$.
 
\section{Bounds for finite number of ports $N$}
\label{finite}
To achieve our goals in this section we have combined in a non-trivial way advanced tools emerging  from statistical analysis. Motivated by Berry-Essen Theorem~\cite{endriu,essen0}, which quantifies the rate of convergence to normal distribution, together with  Central Limit Theorem type reasoning we express probability of success in terms of Gaussian integrals. In particular, starting from expression~\eqref{exp50} in the finite case, we provide useful lower and upper bounds on probability of success of the perfect transmission in MPBT protocol. We can observe, that it closely resembles the expression for the probability distribution function of the distribution of number of successes in $N+1$ Bernoulli trials each with $p=1/2$. We start from proving the following:
\begin{proposition}
\label{prep:bnd}
Let $k=a\sqrt{N}, a\in(0,2).$ Then
    \begin{align}
2\left(\int_0^\infty x^2\frac{1}{\sqrt{2\pi}}\operatorname{e}^{-\frac{(x+a\sqrt{\frac{N}{N+1}})^2}{2}}dx - M(N) - I_1^B(N) -I_2^B(N) \right)-\delta^B(N) \leq p_{succ}
\leq &2\left(\int_0^\infty x^2\frac{1}{\sqrt{2\pi}}\operatorname{e}^{-\frac{(x+a\sqrt{\frac{N}{N+1}})^2}{2}}dx+ M(N)\right), 
\end{align}
where
\begin{equation}
\begin{split}
M(N)&=\frac{2}{\operatorname{e}\sqrt{N+1}}\frac{1}{\sqrt{2\pi}},\quad I_1^B(N)=(N+1)^{-\frac{3}{2}}\frac{1}{\sqrt{2\pi}}, \\ I_2^B(N)&=\left(\frac{\sqrt{N+1}
        }{\sqrt{2\pi}} + 1\right)\operatorname{e}^{-\frac{\left(\sqrt{N+1}-
        1\right)^2}{2}},\\
    \delta^B(N)&=4N^{-1/4}+11\sqrt{N+1}\operatorname{e}^{-\frac{\left(N^{5/8}-1\right)^2}{N+1}}.
\end{split}
\end{equation}
\end{proposition}

The reasoning above applies to finite number of ports. We can see the convergence of these bounds in Figure~\ref{fig:bnd}
\begin{figure}[h]
	\begin{centering}
		\includegraphics[width=1.0\textwidth]{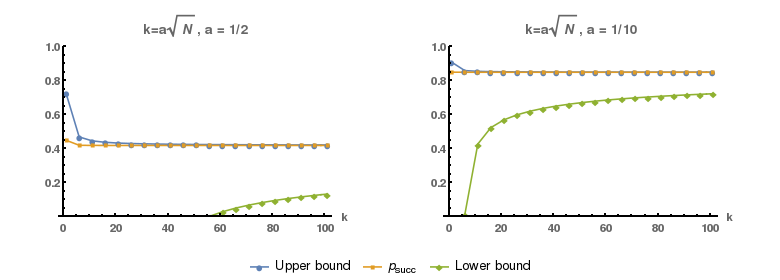}
		\caption{The figures present lower and upper bounds for $p_{succ}$ given in Proposition~\ref{prep:bnd}, for two different values of $a$, where the number of teleported qubits is given by $k~=~a\sqrt{N}$.}
		\label{fig:bnd}%
	\end{centering}
\end{figure}
Taking the limit $N\rightarrow \infty$ we can formulate the following 
\begin{corollary}
\label{col:granica1}
    The limiting value of $p_{succ}$, where $k=a\sqrt{N}, a \in(0,1), N\rightarrow\infty$ is
\be
\lim_{N \to \infty} p_{succ} = 2\int_0^\infty x^2\frac{1}{\sqrt{2\pi}}\operatorname{e}^{-\frac{(x+a)^2}{2}}dx.
\ee
\end{corollary}

When we take $k=o(\sqrt{N})$, i.e. $k/\sqrt{N} \to 0$ when $N\to 0$, we can apply the bounds
(\ref{eq:bnd1})-(\ref{eq:bnd3})
substituting $a\sqrt{\frac{N}{N+1}}$ with $k/ \sqrt{N}$, and when we take the limit $N\to \infty$ we obtain
\begin{corollary}
\label{col:granica2}
    When $k=o(\sqrt{N})$ then the limit of $p_{succ}$ is
    \be
    \lim_{N \to \infty} p_{succ} = 2\int_0^\infty x^2\frac{1}{\sqrt{2\pi}}\operatorname{e}^{-\frac{x^2}{2}}dx=1.
    \ee
 \end{corollary}   

\begin{proof}[Proof of Proposition~\ref{prep:bnd}]
The probability of success in the multi-port teleportation scheme, for even $N$ can be written as:
\be
\label{ps}
p_s=\frac{1}{N+1}\sum_{s=0}^{\frac{N-k}{2}}(2s+1)^2\operatorname{Pr}\left(X=\frac{N-k}{2}-s \right).
\ee
Instead of using discrete probability distribution $\operatorname{Pr}(X\leq l)$ we can approximate it by Gaussian distribution $\Phi\left(\frac{l-\frac{N+1}{2}}{\frac{1}{2}\sqrt{N+1}} \right)$. The error of such estimation is known due to the  Berry-Essen Theorem~\cite{endriu,essen0}, which in our case reads as
\be
\forall \ l=0,\ldots,N+1 \quad \left|\operatorname{Pr}(X\leq l)-\Phi\left(\frac{l-\frac{N+1}{2}}{\frac{1}{2}\sqrt{N+1}} \right)  \right|\leq \frac{4c}{\sqrt{N+1}}.
\ee
The real constant $c$ can be analytically bounded from below and due to~\cite{essen} there is $c\geq \frac{\sqrt{10}+3}{6\sqrt{2\pi}}\approx0.4097$. In this paper for our purposes we set $c=1/2$.
We can rewrite $\Phi\left(\frac{l-\frac{N+1}{2}}{\frac{1}{2}\sqrt{N+1}} \right)$ as
\be
\Phi\left(\frac{l-\frac{N+1}{2}}{\frac{1}{2}\sqrt{N+1}} \right)=\int_{-\infty}^{-\frac{N-1}{2}}G\left( 0,\sigma=\frac{1}{2}\sqrt{N+1};x\right) \operatorname{d}x+\sum_{i=0}^l\int_{-\frac{N+i}{2}}^{-\frac{N+i}{2}+1}G\left(0,\frac{1}{2}\sqrt{N+1};x \right)\operatorname{d}x.
\ee
The first term of the above expression decays to 0 for $N\rightarrow \infty$, more precisely we have
\be
\label{b1}
\begin{split}
\int_{-\infty}^{-\frac{N-1}{2}}G\left( 0,\sigma=\frac{1}{2}\sqrt{N+1};x\right) \operatorname{d}x&=\frac{2}{\sqrt{2\pi(N+1)}} \int_{-\infty}^{-\frac{N-1}{2}}\operatorname{exp}\left(-\frac{2x^2}{N+1}\right)\operatorname{d}x\\
&=\frac{2}{\sqrt{2\pi(N+1)}} \int_{\frac{N-1}{2}}^{\infty}\operatorname{exp}\left(-\frac{2x^2}{N+1}\right)\operatorname{d}x\\
&\leq \sqrt{\frac{N+1}{2}}\operatorname{exp}\left(-\frac{N+1}{2}\right):=G_1(N),
\end{split}
\ee
since the Hoeffding bound states
\be
\label{Hoeff}
\int_r^{\infty}\operatorname{exp}(-\widetilde{c}^2x^2)\operatorname{d}x\leq \frac{\sqrt{\pi}}{\widetilde{c}}\operatorname{exp}\left(-\widetilde{c}^2r^2\right).
\ee 
Notice that factor $G_1(N)$ decays to 0 for $N\rightarrow \infty$.
Now, due to~\eqref{b1}, we have $\forall \ l=0,\ldots,N+1$
\be
\begin{split}
\left|\operatorname{Pr}(X\leq l)-\Phi\left(\frac{l-\frac{N+1}{2}}{\frac{1}{2}\sqrt{N+1}} \right)  \right|&\leq \left|\operatorname{Pr}(X\leq l)-\sum_{i=0}^l\int_{-\frac{N+i}{2}}^{-\frac{N+i}{2}+1}G\left(0,\frac{1}{2}\sqrt{N+1};x \right)\operatorname{d}x \right|+G_1(N)\\
&= \left|\sum_{i=0}^l \operatorname{Pr}(X=i)-\sum_{i=0}^l\int_{-\frac{N+i}{2}}^{-\frac{N+i}{2}+1}G\left(0,\frac{1}{2}\sqrt{N+1};x \right)\operatorname{d}x  \right|+G_1(N)\\ 
&= \sum_{i=0}^l\left|\operatorname{Pr}(X=i)-\int_{-\frac{N+i}{2}}^{-\frac{N+i}{2}+1}G\left(0,\frac{1}{2}\sqrt{N+1};x \right)\operatorname{d}x \right|+G_1(N).
\end{split}
\ee
Getting back to probability of success in~\eqref{ps}:
\be
\label{ps2}
\begin{split}
 \frac{2}{N+1}\sum_{s=0}^{\frac{N-k}{2}}(2s+1)^2 &\Phi_{\sigma}\left(-\frac{k}{2}-s-\frac{1}{2},-\frac{k}{2}-s+\frac{1}{2} \right)\geq \\
p_s&=\frac{2}{N+1}\sum_{s=0}^{\frac{N-k}{2}}(2s+1)^2\operatorname{Pr}\left(X=\frac{N-k}{2}-s \right)\\
&\geq \frac{2}{N+1}\left(\sum_{s=0}^{\frac{N-k}{2}}(2s+1)^2 \Phi_{\sigma}\left(-\frac{k}{2}-s-\frac{1}{2},-\frac{k}{2}-s+\frac{1}{2} \right)-\sum_{s=0}^{\frac{N-k}{2}}(2s+1)^2\delta_s \right), 
\end{split}
\ee
where 
\be
\delta_s=\left|\operatorname{Pr}\left(x=\frac{N-k}{2}-s\right)-\Phi_{\sigma}\left(-\frac{k}{2}-s-\frac{1}{2},-\frac{k}{2}-s+\frac{1}{2} \right)   \right|,\qquad \sigma=\frac{1}{2}\sqrt{N+1}. 
\ee
The upper bound comes from the Lemma 7.2 from \cite{phdthesis}. Since $X\sim B(N+1,\frac{1}{2})$
\be
\Phi_{\sigma}\left(-\frac{k}{2}-s-\frac{1}{2},-\frac{k}{2}-s+\frac{1}{2} \right)\geq \operatorname{Pr}\left(X=\frac{N-k}{2}-s \right)
\ee
if
\be
\frac{N+1}{2}+1 \leq N+1 - \left(\frac{N-k}{2}-s\right) \ee 
which is equivalent to $k+2s\geq1$ which clearly holds for $k$ and $s$ in question.

At this point we can divide further considerations into two main steps:
\begin{enumerate}
	\item Showing that $\frac{2}{N+1}\sum_{s=0}^{\frac{N-k}{2}}(2s+1)^2\delta_s\rightarrow 0$ for $N\rightarrow \infty$ in~\eqref{ps2}.
	\item Evaluating the first term $ \frac{2}{N+1}\sum_{s=0}^{\frac{N-k}{2}}(2s+1)^2 \Phi_{\sigma}\left(-\frac{k}{2}-s-\frac{1}{2},-\frac{k}{2}-s+\frac{1}{2} \right)$ from  expression~\eqref{ps2} in the asymptotic regime $N\rightarrow \infty$.
\end{enumerate}
{\bf Showing decaying of the second term in~\eqref{ps2}.}\\ Let us focus on the first point. We write
\be
\label{lhs}
\begin{split}
\frac{2}{N+1}\sum_{s=0}^{\frac{N-k}{2}}(2s+1)^2\delta_s=\underbrace{\frac{2}{N+1}\sum_{s=0}^{m(N)}(2s+1)^2\delta_s}_A +\underbrace{\frac{2}{N+1}\sum_{s=m(N)}^{\frac{N-k}{2}}(2s+1)^2\delta_s}_B.
\end{split}
\ee
Substituting $m(N)=(1/2)( b(N)\sqrt{N}-1)$, using fact $\sum_{s=0}^{\frac{N-k}{2}}\delta_s\leq \frac{4c}{\sqrt{N+1}}$ with $c=1/2$, we bound $A$  as
\be
\label{bound_A}
A\leq \frac{2(2m(N)+1)^2}{N+1}\sum_{s=0}^{m(N)}\delta_s\leq \frac{4b^2(N)N}{(N+1)^{3/2}}\leq \frac{4Nb^2(N)}{N^{3/2}}=4N^{-1/4},
\ee
where we choose for example $b(N)=N^{1/8}$ getting demanded asymptotic. For the factor $B$, first we write
\be
\label{boundy}
\begin{split}
\delta_s&\leq \operatorname{Pr}\left(x=\frac{N-k}{2}-s \right)+\Phi_{\sigma}\left(-\frac{k}{2}-s-\frac{1}{2},-\frac{k}{2}-s+\frac{1}{2}  \right)\\
&\leq \operatorname{Pr}\left(x\leq \frac{N-k}{2}-s \right)+\Phi_{\sigma}\left(-\frac{k}{2}-s+\frac{1}{2}  \right)\\
&\leq \operatorname{exp}\left(-\frac{(k+2s+1)^2}{2(N+1)}\right)+ \operatorname{exp}\left(-\frac{(k+2s-1)^2}{2(N+1)}\right)\\
&\leq 2\operatorname{exp}\left(-\frac{(k+2s-1)^2}{2(N+1)}\right).
\end{split}
\ee
Indeed, for $\Phi_{\sigma}\left(-\frac{k}{2}-s+\frac{1}{2}  \right)$ we use $\sigma=\frac{1}{2}\sqrt{N+1}$ and the Hoeffding bound~\eqref{Hoeff}:
\be
\begin{split}
\Phi_{\sigma}\left(-\frac{k}{2}-s+\frac{1}{2} \right)=\frac{1}{\sqrt{2\pi}\sigma}\int_{-\frac{k}{2}-s+\frac{1}{2}}^{\infty}\operatorname{exp}\left(-\frac{1}{2}\frac{x^2}{\sigma^2}\right)\operatorname{d}x&\leq \operatorname{exp}\left(-\frac{1}{2}\frac{(-\frac{k}{2}-s+\frac{1}{2})^2}{\sigma^2}\right)\\
&= \operatorname{exp}\left(-\frac{(k+2s-1)^2}{2(N+1)}\right),
\end{split}
\ee
For $\operatorname{Pr}\left(x\leq \frac{N-k}{2}-s \right)$ we use the Chernoff bound $\operatorname{Pr}(X<(1-\delta )\mu )\leq \operatorname{e}^{-\frac{\mu \delta ^{2}}{2}}$, with $\mu=N+1$ and $1-\delta=\frac{N-k-2s}{N+1}$:
\be
\operatorname{Pr}\left(x\leq \frac{N-k}{2}-s \right)\leq \operatorname{exp}\left(-\frac{(k+2s+1)^2}{2(N+1)}\right).
\ee
Applying~\eqref{boundy} to $B$ we reduce to
\be
\label{B2}
\begin{split}
B\leq \frac{4}{N+1}\sum_{s=m(N)}^{\frac{N-k}{2}}(2s+1)^2\operatorname{exp}\left(-\frac{(k+2s-1)^2}{2(N+1)} \right)&\leq \frac{4(N-k+1)}{N+1}\sum_{s=m(N)}^{\frac{N-k}{2}}\operatorname{exp}\left(-\frac{(k+2s-1)^2}{2(N+1)} \right)\\
&\leq \frac{4(N-k+1)}{N+1}\sum_{s=m(N)}^{\infty}\operatorname{exp}\left(-\frac{(k+2s-1)^2}{2(N+1)} \right)\\
&\leq 4\sum_{s=m(N)}^{\infty}\operatorname{exp}\left(-\frac{(k+2s-1)^2}{2(N+1)} \right)\\
&\leq \frac{11}{\sqrt{N+1}}\operatorname{exp}\left(-\frac{(k+2m(N)-1)^2}{2(N+1)}\right)\\
&\leq 11\sqrt{N+1}\operatorname{exp}\left(-\frac{2m^2(N)}{N+1}\right),
\end{split}
\ee
where in the fourth inequality we use the Hoeffding inequality. Finally, using explicit form of the function $m(N)=(1/2)( N^{5/8}-1)$ and expression~\eqref{bound_A} we bound left-hand side of~\eqref{lhs} as
\be
\label{eq:delta_bound}
\frac{2}{N+1}\sum_{s=0}^{\frac{N-k}{2}}(2s+1)^2\delta_s\leq 4N^{-3/2}+11\sqrt{N+1}\operatorname{exp}\left(-\frac{\left(N^{5/8}-1\right)^2}{N+1}\right)
\ee
which clearly decays to 0 with $N\rightarrow \infty$. It means $p_s$ in~\eqref{ps2} can be bounded from the below 
\be
\begin{split}
p_s&\geq \frac{2}{N+1}\sum_{s=0}^{\frac{N-k}{2}}(2s+1)^2 \Phi_{\sigma}\left(-\frac{k}{2}-s-\frac{1}{2},-\frac{k}{2}-s+\frac{1}{2} \right)-4N^{-1/4}-11\sqrt{N+1}\operatorname{exp}\left(-\frac{\left(N^{5/8}-1\right)^2}{N+1}\right).
\end{split}
\ee
{\bf Estimation of the first term in~\eqref{ps2}.}
We now want to estimate the expression
\be
\label{eq:paski}
\frac{2}{N+1}\left(\sum_{s=0}^{\frac{N-k}{2}}(2s+1)^2
\Phi_{\sigma}\left(-\frac{k}{2}-s-\frac{1}{2},-\frac{k}{2}-s+\frac{1}{2} \right)\right).
\ee
First, we can write
\begin{align}
&\frac{2}{N+1}\left(\sum_{s=0}^{\frac{N-k}{2}}(2s+1)^2
\Phi_{\sigma}\left(-\frac{k}{2}-s-\frac{1}{2},-\frac{k}{2}-s+ \frac{1}{2},\right)\right)  = 
\frac{2}{N+1}\sum_{s=0}^{\frac{N-k}{2}}(2s+1)^2
\int_{-\frac{k}{2}-s-\frac{1}{2}}^{-\frac{k}{2}-s+\frac{1}{2}} \frac{2}{\sqrt{2\pi(N+1)}}\operatorname{e}^{-\frac{2x^2}{N+1}} dx\\
  &\geq  \frac{2}{N+1}\sum_{s=0}^{\frac{N-k}{2}}(2s+1)^2
    \frac{2}{\sqrt{2\pi(N+1)}}\operatorname{e}^{-\frac{(k+2s+1)^2}{2(N+1)}}\\
  &= 2 \frac{2}{\sqrt{N+1}}\sum_{s=0}^{\frac{N-k}{2}}\left(\frac{1}{\sqrt{N+1}} + s \frac{2}{\sqrt{N+1}}\right)^2
    \frac{1}{\sqrt{2\pi}}\operatorname{e}^{-\frac{\left(\frac{1}{\sqrt{N+1}}+s\frac{2}{\sqrt{N+1}} +\frac{k}{\sqrt{N+1}} \right)^2}{2}}.
\end{align}
where we replace integral over interval $[{-\frac{k}{2}-s-\frac{1}{2}},{-\frac{k}{2}-s+\frac{1}{2}}]$ with local minimum of integrated function multiplied by the length of the interval, which is $1$.
One can identify the last expression as a Riemann left-sum of function
\be f^N(z)\coloneqq z^2\operatorname{e}^{\frac{(z+\frac{k}{\sqrt{N+1}})^2}{2}}\ee on the
interval $[\frac{1}{\sqrt{N+1}}, \frac{N-k}{\sqrt{N+1}}]$. We will transform this sum in the manner depicted in Figure \ref{fig:Riemann}. 
Setting $k=a\sqrt{N}$ the target function takes form
\be
f^N(z)=z^2\operatorname{e}^{\frac{(z+a\sqrt{\frac{N}{N+1}})^2}{2}}.
\ee
 One can check, that for
\be
0\leq z\leq m \coloneqq \frac{-\sqrt{\frac{N}{N+1}}a+\sqrt{\frac{N}{N+1}a^2+8}}{2}
\ee it is monotonically increasing, while for
$z\geq m $ it decreases.
Therefore it can be split into two components: one consisting of partitions $\{[\frac{1}{\sqrt{N+1}}, \frac{1+2s}{\sqrt{N+1}}],
s\in 0,\dots,s_m\}$
where $\frac{1+2s_m}{\sqrt{N+1}} \leq m$, and the other consisting of partitions for higher values of $s$. Thus we have
\begin{align}
     \frac{2}{\sqrt{N+1}}\sum_{s=0}^{\frac{N-k}{2}}&\left(\frac{1}{\sqrt{N+1}} + s \frac{2}{\sqrt{N+1}}\right)^2
    \frac{1}{\sqrt{2\pi}}\operatorname{e}^{-\frac{\left(\frac{1}{\sqrt{N+1}}+s\frac{2}{\sqrt{N+1}} +\frac{k}{\sqrt{N+1}}
    \right)^2}{2}} \\
     = &2 \frac{2}{\sqrt{N+1}}\sum_{s=0}^{s_m}\left(\frac{1}{\sqrt{N+1}} + s \frac{2}{\sqrt{N+1}}\right)^2
    \frac{1}{\sqrt{2\pi}}\operatorname{e}^{-\frac{\left(\frac{1}{\sqrt{N+1}}+s\frac{2}{\sqrt{N+1}} +\frac{k}{\sqrt{N+1}}
    \right)^2}{2}}\\
    &+  2 \frac{2}{\sqrt{N+1}}\sum_{s=s_m+1}^{\frac{N-k}{2}}\left(\frac{1}{\sqrt{N+1}} + s \frac{2}{\sqrt{N+1}}\right)^2
    \frac{1}{\sqrt{2\pi}}\operatorname{e}^{-\frac{\left(\frac{1}{\sqrt{N+1}}+s\frac{2}{\sqrt{N+1}} +\frac{k}{\sqrt{N+1}}
    \right)^2}{2}}\\
    =& S_L^N\left(\frac{1}{\sqrt{N+1}}, \frac{1+2(s_m+1)}{\sqrt{N+1}}\right) + S_L^N\left(\frac{1+2(s_m+1)}{\sqrt{N+1}},
    \frac{N-a\sqrt{N}+1}{\sqrt{N+1}}\right),
\end{align}
where $S_L^N\left(\frac{1}{\sqrt{N+1}}, \frac{1+2(s_m+1)}{\sqrt{N+1}}\right)$ denotes left Riemann sum on the interval
$[\frac{1}{\sqrt{N+1}}, \frac{1+(2s_m+1)}{\sqrt{N+1}}]$. Setting $t=s+1$, we can write
\begin{align}
    S_L^N\left(\frac{1}{\sqrt{N+1}}, \frac{2s_m+1}{\sqrt{N+1}}\right)&=2 \frac{2}{\sqrt{N+1}}\sum_{s=0}^{s_m}\left(\frac{1}{\sqrt{N+1}} + s \frac{2}{\sqrt{N+1}}\right)^2
    \frac{1}{\sqrt{2\pi}}\operatorname{e}^{-\frac{\left(\frac{1}{\sqrt{N+1}}+s\frac{2}{\sqrt{N+1}} +\frac{k}{\sqrt{N+1}}
    \right)^2}{2}} \\
    &=2 \frac{2}{\sqrt{N+1}}\sum_{t=1}^{t_m}\left(-\frac{1}{\sqrt{N+1}} + t \frac{2}{\sqrt{N+1}}\right)^2
    \frac{1}{\sqrt{2\pi}}\operatorname{e}^{-\frac{\left(-\frac{1}{\sqrt{N+1}}+t\frac{2}{\sqrt{N+1}} +\frac{k}{\sqrt{N+1}} \right)^2}{2}}\\
    &=S_R^N\left(-\frac{1}{\sqrt{N+1}}, \frac{-1+t_m}{\sqrt{N+1}}\right). 
\end{align}
Thus
\begin{align}
    &S_L^N\left(\frac{1}{\sqrt{N+1}}, \frac{1+2(s_m+1)}{\sqrt{N+1}}\right) + S_L^N\left(\frac{1+2(s_m+1)}{\sqrt{N+1}},
    \frac{N-a\sqrt{N}+1}{\sqrt{N+1}}\right) \label{eq:R1} \\
    &= S_R^N\left(-\frac{1}{\sqrt{N+1}}, \frac{-1+2t_m}{\sqrt{N+1}}\right) + S_L^N\left(\frac{1+2(s_m+1)}{\sqrt{N+1}},
    \frac{N-a\sqrt{N}+1}{\sqrt{N+1}}\right)\label{eq:R2}.
\end{align}
We can observe, that we shifted the first left sum by $-\frac{2}{\sqrt{N+1}}$, which corresponds to single partition
interval, turning it into right sum. Supplementing the upper bound of the "missing middle
term" $\frac{2}{\sqrt{N+1}}m^2 \frac{1}{\sqrt{2\pi}}\operatorname{e}^{-\frac{\left(m+\frac{k}{\sqrt{N+1}} \right)^2}{2}}$, we obtain the
following restriction
\begin{align}
S_R^N\left(-\frac{1}{\sqrt{N+1}}, \frac{-1+2t_m}{\sqrt{N+1}}\right) + S_L^N\left(\frac{1+2(s_m+1)}{\sqrt{N+1}},
    \frac{N-a\sqrt{N}+1}{\sqrt{N+1}}\right) + 
\frac{2}{\sqrt{N+1}}m^2
    \frac{1}{\sqrt{2\pi}}\operatorname{e}^{-\frac{\left(m+\frac{k}{\sqrt{N+1}} \right)^2}{2}}dx \label{eq:R3}\\
    \geq  
   \int_{-\frac{1}{\sqrt{N+1}}}^\frac{N-a\sqrt{N}+1}{\sqrt{N+1}}x^2\frac{1}{\sqrt{2\pi}}\operatorname{e}^{-\frac{(x+a\sqrt{\frac{N}{N+1}})^2}{2}}dx\\
    \geq  
   \int_{\frac{1}{\sqrt{N+1}}}^\frac{N-a\sqrt{N}+1}{\sqrt{N+1}}x^2\frac{1}{\sqrt{2\pi}}\operatorname{e}^{-\frac{(x+a\sqrt{\frac{N}{N+1}})^2}{2}}dx.
\end{align}
These three sums are depicted in Figure \ref{fig:Riemann}.
\begin{figure}[h]
    \centering
    \includegraphics[width=\textwidth]{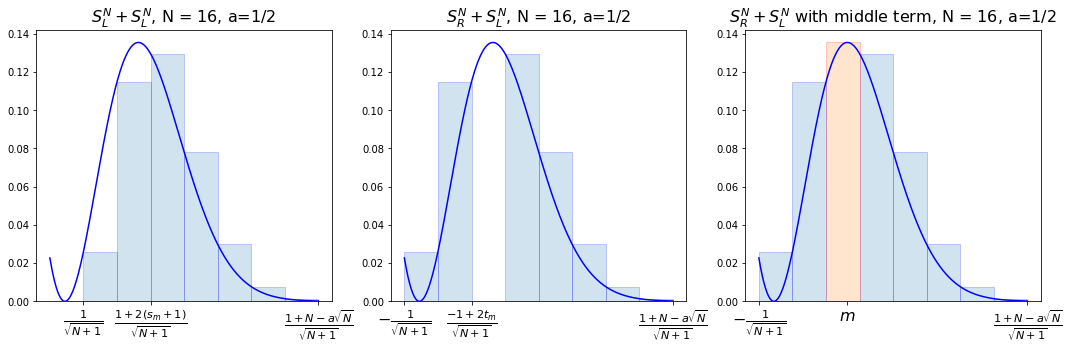}
    \caption{Three Riemann sums: (\ref{eq:R1}), (\ref{eq:R2}) and (\ref{eq:R3}).}
    \label{fig:Riemann}
\end{figure}

Thus the initial Riemann sum (\ref{eq:R1}) is bounded by
\begin{align}
     \frac{2}{\sqrt{N+1}} \sum_{s=0}^{\frac{N-k}{2}}&\left(\frac{1}{\sqrt{N+1}} + s \frac{2}{\sqrt{N+1}}\right)^2
    \frac{1}{\sqrt{2\pi}}\operatorname{e}^{-\frac{\left(\frac{1}{\sqrt{N+1}}+s\frac{2}{\sqrt{N+1}} +\frac{k}{\sqrt{N+1}}
    \right)^2}{2}} \\
    &= S_L^N\left(\frac{1}{\sqrt{N+1}}, \frac{2s_m+3}{\sqrt{N+1}}) + S_L^N(\frac{2s_m+3}{\sqrt{N+1}},
    \frac{N-a\sqrt{N}+1}{\sqrt{N+1}}\right) \\
    & = S_R^N\left(-\frac{1}{\sqrt{N+1}}, \frac{2s_m+1}{\sqrt{N+1}}) + S_L^N(\frac{2s_m+3}{\sqrt{N+1}},
    \frac{N-a\sqrt{N}+1}{\sqrt{N+1}}\right)\\
    & \geq
    \int_{\frac{1}{\sqrt{N+1}}}^\frac{N-a\sqrt{N}+1}{\sqrt{N+1}}x^2\frac{1}{\sqrt{2\pi}}\operatorname{e}^{-\frac{(x+a\sqrt{\frac{N}{N+1}})^2}{2}}dx
    - \frac{2}{\sqrt{N+1}}m^2
    \frac{1}{\sqrt{2\pi}}\operatorname{e}^{\frac{\left(m+\frac{k}{\sqrt{N+1}} \right)^2}{2}} \\
& \geq
    \int_{\frac{1}{\sqrt{N+1}}}^\frac{N-a\sqrt{N}+1}{\sqrt{N+1}}x^2\frac{1}{\sqrt{2\pi}}\operatorname{e}^{-\frac{(x+a\sqrt{\frac{N}{N+1}})^2}{2}}dx
    - \frac{2}{\sqrt{N+1}}m^2
    \frac{1}{\sqrt{2\pi}}\operatorname{e}^{\frac{m^2}{2}} \\
& \geq
    \int_{\frac{1}{\sqrt{N+1}}}^\frac{N-a\sqrt{N}+1}{\sqrt{N+1}}x^2\frac{1}{\sqrt{2\pi}}\operatorname{e}^{-\frac{(x+a\sqrt{\frac{N}{N+1}})^2}{2}}dx
    - \frac{4}{\operatorname{e}\sqrt{N+1}}
    \frac{1}{\sqrt{2\pi}}\\
& \geq
    \int_{\frac{1}{\sqrt{N+1}}}^\frac{N-a\sqrt{N}+1}{\sqrt{N+1}}x^2\frac{1}{\sqrt{2\pi}}\operatorname{e}^{-\frac{(x+a\sqrt{\frac{N}{N+1}})^2}{2}}dx
    - \frac{2}{\sqrt{N+1}}
    \frac{1}{\sqrt{2\pi}}
\end{align}
and thus (\ref{eq:paski}) admits the following restriction
\begin{align}
\frac{2}{N+1}&\left(\sum_{s=0}^{\frac{N-k}{2}}(2s+1)^2\Phi_{\sigma}\left(-\frac{k}{2}-s-\frac{1}{2},-\frac{k}{2}-s+\frac{1}{2} \right)\right)\\
    & \geq    2 \frac{2}{\sqrt{N+1}} \sum_{s=0}^{\frac{N-k}{2}}\left(\frac{1}{\sqrt{N+1}} + s \frac{2}{\sqrt{N+1}}\right)^2
    \frac{1}{\sqrt{2\pi}}\operatorname{e}^{-\frac{\left(\frac{1}{\sqrt{N+1}}+s\frac{2}{\sqrt{N+1}} +\frac{k}{\sqrt{N+1}}
    \right)^2}{2}} \\
    &\geq
    2\left(\int_{\frac{1}{\sqrt{N+1}}}^\frac{N-a\sqrt{N}+1}{\sqrt{N+1}}x^2\frac{1}{\sqrt{2\pi}}\operatorname{e}^{-\frac{(x+a\sqrt{\frac{N}{N+1}})^2}{2}}dx
    - 
     \frac{2}{\sqrt{N+1}}
    \frac{1}{\sqrt{2\pi}}\right).
\end{align}
Equivalent reasoning can be used to obtain the upper bound of (\ref{eq:paski}).
Thus, for finite $N$ one obtains
\begin{align}
\displaystyle2\int_{\frac{1}{\sqrt{N+1}}}^\frac{N-a\sqrt{N}+1}{\sqrt{N+1}}x^2\frac{1}{\sqrt{2\pi}}\operatorname{e}^{-\frac{(x+a\sqrt{\frac{N}{N+1}})^2}{2}}dx+ 2M \geq p_{succ} \geq 
    2\int_{\frac{1}{\sqrt{N+1}}}^\frac{N-a\sqrt{N}+1}{\sqrt{N+1}}x^2\frac{1}{\sqrt{2\pi}}\operatorname{e}^{-\frac{(x+a\sqrt{\frac{N}{N+1}})^2}{2}}dx
    -2M-\sum_{s=0}^{\frac{N-k}{2}}(2s+1)^2\delta_s,
\end{align}
where $M = \frac{2}{\sqrt{N+1}}\frac{1}{\sqrt{2\pi}}$. This bound can be further refined. On the one hand, clearly
\be
\displaystyle\int_{\frac{1}{\sqrt{N+1}}}^\frac{N-a\sqrt{N}+1}{\sqrt{N+1}}x^2\frac{1}{\sqrt{2\pi}}\operatorname{e}^{-\frac{(x+a\sqrt{\frac{N}{N+1}})^2}{2}}dx+ M \leq \displaystyle\int_0^\infty x^2\frac{1}{\sqrt{2\pi}}\operatorname{e}^{-\frac{(x+a\sqrt{\frac{N}{N+1}})^2}{2}}dx+ M.
\ee
On the other hand
\be
\begin{split}
    &\int_{\frac{1}{\sqrt{N+1}}}^\frac{N-a\sqrt{N}+1}{\sqrt{N+1}} x^2\frac{1}{\sqrt{2\pi}}\operatorname{e}^{-\frac{(x+a\sqrt{\frac{N}{N+1}})^2}{2}}dx
    -M-\sum_{s=0}^{\frac{N-k}{2}}(2s+1)^2\delta_s = \\
    &\int_0^\infty x^2\frac{1}{\sqrt{2\pi}}\operatorname{e}^{-\frac{(x+a\sqrt{\frac{N}{N+1}})^2}{2}}dx
    -M-\sum_{s=0}^{\frac{N-k}{2}}(2s+1)^2\delta_s - \int_0^{\frac{1}{\sqrt{N+1}}}x^2\frac{1}{\sqrt{2\pi}}\operatorname{e}^{-\frac{(x+a\sqrt{\frac{N}{N+1}})^2}{2}}dx - \int_\frac{N-a\sqrt{N}+1}{\sqrt{N+1}}^\infty x^2\frac{1}{\sqrt{2\pi}}\operatorname{e}^{-\frac{(x+a\sqrt{\frac{N}{N+1}})^2}{2}}dx=\\
&\int_0^\infty x^2\frac{1}{\sqrt{2\pi}}\operatorname{e}^{-\frac{(x+a\sqrt{\frac{N}{N+1}})^2}{2}}dx
    -M-\sum_{s=0}^{\frac{N-k}{2}}(2s+1)^2\delta_s - I_1-I_2.
\end{split}
\ee
One has
\be
I_1 = \int_0^{\frac{1}{\sqrt{N+1}}}x^2\frac{1}{\sqrt{2\pi}}\operatorname{e}^{-\frac{(x+a\sqrt{\frac{N}{N+1}})^2}{2}}dx \leq \int_0^{\frac{1}{\sqrt{N+1}}}x^2\frac{1}{\sqrt{2\pi}}\operatorname{e}^{-\frac{x^2}{2}}dx\leq \int_0^{\frac{1}{\sqrt{N+1}}} \frac{1}{N+1}\frac{1}{\sqrt{2\pi}}dx = (N+1)^{-\frac{3}{2}}\frac{1}{\sqrt{2\pi}}
\ee
and having definition of  the complementary error function $\operatorname{erfc}(x):=1-\operatorname{erf}(x)$, where $\operatorname{erf}(x):=\frac{2}{\sqrt{\pi}}\int_0^x \operatorname{e}^{-t^2}dt$ is the error function we write
\be
    \begin{split}
I_2 = \int_\frac{N-a\sqrt{N}+1}{\sqrt{N+1}}^\infty x^2\frac{1}{\sqrt{2\pi}}\operatorname{e}^{-\frac{(x+a\sqrt{\frac{N}{N+1}})^2}{2}}dx
    \leq  \int_\frac{N-a\sqrt{N}+1}{\sqrt{N+1}}^\infty x^2\frac{1}{\sqrt{2\pi}}\operatorname{e}^{-\frac{x^2}{2}}dx \coloneqq
    \int_b^\infty x^2\frac{1}{\sqrt{2\pi}}\operatorname{e}^{-\frac{x^2}{2}}dx = \\
        \frac{b \operatorname{e}^ {-\frac{b^2}{2}}}{\sqrt{2\pi}} + \frac{1}{2}\operatorname{erfc}\left(\frac{b}{\sqrt{2}}\right) =
        \frac{b \operatorname{e}^ {-\frac{b^2}{2}}}{\sqrt{2\pi}} + \frac{1}{\sqrt{\pi}}\int_\frac{b}{\sqrt{2}}^\infty
        \operatorname{e}^{-\frac{x^2}{2}}dx \leq \frac{b \operatorname{e}^ {-\frac{b^2}{2}}}{\sqrt{2\pi}} + \operatorname{e}^{-\frac{b^2}{2}}, \label{eq:Hoeffiding}\\
    \end{split}
\ee
    where in (\ref{eq:Hoeffiding}) we use Hoeffiding bound, i.e. $\int_t^\infty \operatorname{e}^{-cx^2} dx \leq \frac{\sqrt{\pi}}{c}
    \operatorname{e}^{-c^2t^2}$ and $b=\frac{N+1-a\sqrt{N}}{\sqrt{N+1}}=\sqrt{N+1}- a\sqrt{\frac{N}{N+1}}$. Furthermore
\be
    \begin{split}
        \frac{b \operatorname{e}^ {-\frac{b^2}{2}}}{\sqrt{2\pi}} + \operatorname{e}^{-\frac{b^2}{2}} = \left(\frac{\left(\sqrt{N+1}-
        a\sqrt{\frac{N}{N+1}}\right)}{\sqrt{2\pi}} + 1\right)\operatorname{e}^{-\frac{\left(\sqrt{N+1}-
        a\sqrt{\frac{N}{N+1}}\right)^2}{2}} \leq \\ \left(\frac{\left(\sqrt{N+1}-
        a\sqrt{\frac{N}{N+1}}\right)}{\sqrt{2\pi}} + 1\right)\operatorname{e}^{-\frac{\left(\sqrt{N+1}-
        1\right)^2}{2}} \leq \left(\frac{\sqrt{N+1}
        }{\sqrt{2\pi}} + 1\right)\operatorname{e}^{-\frac{\left(\sqrt{N+1}-
        1\right)^2}{2}}.
    \end{split}
\ee

Finally, taking into account (\ref{eq:delta_bound}), we can for a given $N$ and $k=a\sqrt{N}, a\in (0,1)$ formulate the following lower and upper bounds for $p_{succ}$:


\begin{align}
&2\left(\int_0^\infty x^2\frac{1}{\sqrt{2\pi}}\operatorname{e}^{-\frac{(x+a\sqrt{\frac{N}{N+1}})^2}{2}}dx+ \frac{2}{\operatorname{e}\sqrt{N+1}}\frac{1}{\sqrt{2\pi}}\right) \geq p_{succ} \geq\label{eq:bnd1}\\
 &2\left(\int_0^\infty x^2\frac{1}{\sqrt{2\pi}}\operatorname{e}^{-\frac{(x+a\sqrt{\frac{N}{N+1}})^2}{2}}dx - \frac{2}{\operatorname{e}\sqrt{N+1}}\frac{1}{\sqrt{2\pi}} - \left(\frac{\sqrt{N+1}
        }{\sqrt{2\pi}} + 1\right)\operatorname{e}^{-\frac{\left(\sqrt{N+1}-
        1\right)^2}{2}} - (N+1)^{-\frac{3}{2}}\frac{1}{\sqrt{2\pi}}\right)+\\ &-\frac{2}{N+1}\left(4N^{-1/4}+11\sqrt{N+1}\operatorname{e}^{-\frac{\left(N^{5/8}-1\right)^2}{N+1}}\right).\label{eq:bnd3}
\end{align}
\end{proof}

\bibliographystyle{plainnat}
\bibliography{biblio}
\end{document}